\documentclass[a4paper,11pt]{article}

\usepackage[utf8]{inputenc} %[utf8] [latin1] [applemac]
\usepackage[T1]{fontenc}
\usepackage[english]{babel}
\usepackage[utopia,cal]{mathdesign}\usepackage{amsmath,amsthm} % also charter, garamond; ! replaces amssymb
\usepackage{erewhon}

\usepackage{natbib}
\usepackage{xcolor}
\usepackage{setspace}%\singlespacing%\doublespacing%\setstretch{1.1}
\usepackage{enumitem}\setlist[enumerate]{label={\rm(\roman*)}}

\usepackage{mathtools}
\usepackage{dsfont} % provide \mathds, alternative to \mathbb 

\theoremstyle{plain}
\newtheorem{theorem}{Theorem}%[section]

\newtheorem{lemma}{Lemma}

\theoremstyle{definition}
\newtheorem{definition}{Definition}

\newtheoremstyle{key}{3pt}{3pt}{}{}{\itshape}{:}{.5em}{}
\theoremstyle{key}
\newtheorem*{JEL}{JEL subject classification}
\newtheorem*{keywords}{Keywords}

\newcommand{\abs}[1]{\lvert#1\rvert}
\newcommand{\absd}[1]{\left\lvert#1\right\rvert}

\newcommand{\indi}[1]{\mathbf{1}_{\{#1\}}}

\newcommand{\R}{\mathbb{R}}

\newcommand{\A}{{\mathcal A}}

\newcommand{\F}{{\mathcal F}}

\newcommand{\cA}{{\cal A}}

\newcounter{enum+} % to interrupt enumerate manually
\newcommand{\comment}[1]{}

\title{\textbf{Strategic Irreversible Investment}\thanks{Financial support by the German Research Foundation (DFG) [CRC 1283/2 2021 – 317210226] is gratefully acknowledged.}}
\author{Jan-Henrik Steg\thanks{Center for Mathematical Economics, Bielefeld University, Germany. E-mail: jsteg@uni-bielefeld.de.}}

\date{October 19, 2024}

\begin{document}

\maketitle

\begin{abstract}
This paper studies oligopolistic irreversible investment with closed-loop strategies. These permit fully dynamic interactions that result in much richer strategic behavior than previous studies with open-loop strategies allow. The tradeoff between preemption incentives and the option value of waiting becomes distinctly visible. Strategies that depend on present capital stocks enable credible reactions that deter from excessive preemption and support positive option values in equilibrium. Simpler strategies lead into a ``preemption trap'' with perfectly competitive outcome and zero net present values. To obtain these results, a novel concept of Markov perfect equilibrium is developed that copes with optimal investment taking the form of singular control.
\end{abstract}

\begin{JEL}
C61, C73, D25, D43, G31, L11, L13
\end{JEL}

\begin{keywords}
Irreversible investment, Oligopoly, Markov perfect equilibrium, Singular control
\end{keywords}

\section{Introduction}

Investment in real assets like production capacity has a strong strategic value in competitive environments, because it constitutes a credible commitment. But long-term commitments are also risky, as future returns depend on many uncertain factors. Trading off these implications is crucial for deciding how much to invest at which points in time—i.e., for assessing which commitments are in fact valuable. Thus, it is already well understood that the real options approach to investment, which focuses on uncertainty that cannot be influenced, needs to be enriched by a game theoretic perspective to analyze the strategic consequences.

However, attempts to do so have been restricted by the fact that the familiar game theoretic concepts cannot easily be applied to the standard setting of real options theory. The real option models used in both theory and practice are formulated in continuous time, because then one has powerful tools that reduce the problems of determining the option value and the optimal investment strategy to solving deterministic differential equations. This permits to consider many different features of real options, whereas modeling highly dynamic strategic interactions quickly leads to conceptual problems or intractability.

As a result, most studies of strategic real options consider only a single investment opportunity for each firm. This is enough to show that preemption can significantly reduce the option value of waiting, but this also strongly limits strategic reactions---once a firm has invested, the other is left with a pure optimization problem.\footnote{There is a large literature on real option games with indivisible investment, which includes, to give just some typical examples, \cite{Grenadier96}, \cite{Weeds02}, and \cite{Bustamante15}. Many of these papers obtain similar results as \cite{FudenbergTirole85} did in a deterministic setting.}
A smaller strand of the literature has considered divisible investments. In these models, however, the firms are restricted to using open-loop strategies, which do not react to deviations from planned investment.\footnote{Divisible investment with open-loop strategies is studied by, e.g., \cite{Baldursson98}; \cite{Grenadier02}; \cite{Aguerrevere03}; \cite{Aguerrevere09}; \cite{BackPaulsen09}; \cite{Steg12}; and \cite{MorellecZhdanov19}.} 
In the case of capacity investment, the game is then effectively a one-shot Cournot competition, where each firm commits to an entire investment path instead of a single investment size. This rules out strategic moves like any preemptive investment by assumption, because the other firms plainly cannot react---they must follow their initial plans also when these are no longer optimal. Such models seem unsatisfactory for many applications.\footnote{The plausibility of precommitment strategies was already challenged by \cite{Spence79} in the context of a deterministic capital accumulation game and also by \cite{FudenbergTirole85} in the context of a technology adoption timing game.}

This paper studies divisible investments with closed-loop strategies. At each point in time, the firms can make arbitrary investments to increase their capital stocks. The strategies allow the firms to condition their investments not only on the evolution of the aggregate uncertainty factor but also on the actual capital stocks. This permits the study of fully dynamic interactions. As a result, a novel and much richer strategic behavior can be observed. In equilibrium, any additional investment would also increase opponent investment, which discourages any firm from investing more than planned. These reactions are credible and gradual, because the strategies depend only on the current state---the equilibrium is Markov perfect. There is no threat to switch to an extreme punishment regime, which would require to keep track of the whole investment path of every firm. 

This new type of equilibrium investment exists with different levels to which preemption is moderated. Hence, there is also a range for the option values that remain in equilibrium. More aggressive investment implies lower equilibrium payoffs. In the limit, the net present value of any investment is zero—like under perfect competition. This represents extreme preemption in consequence of the divisibility of investment and is the closed-loop equilibrium with the most simple form. But there is also a limit to how permissive the more differentiating strategies can be in equilibrium, which then implies an upper bound for the payoffs that can be sustained.

These equilibrium strategies are in sharp contrast to the open-loop equilibrium. There, since the firms are assumed not to react to actual investments, it is not possible to compete for any option values and, thus, also not necessary to sustain them by particular strategies. In fact, it is even optimal to behave myopically and ignore all future investments.\footnote{The optimality of myopic behavior was first observed by \cite{Leahy93} in a model of perfect competition and then reinforced by \cite{BaldurssonKaratzas97}. It was transferred to an oligopoly model by \cite{Grenadier02} and \cite{BackPaulsen09}, and a general proof was given by \cite{Steg12}.}
The only relation between the firms’ investment decisions in open-loop equilibrium is the same as in static Cournot competition, and this alone determines option values and profits.

Establishing closed-loop equilibria requires to overcome the conceptual difficulties discussed by \cite{BackPaulsen09}---which open-loop strategies simply evade by the lack of interaction. Therefore, this paper develops an appropriate concept of Markov perfect equilibrium for the relevant class of games. It allows the firms to follow the kind of investment paths that are optimal in all related models of divisible investment under uncertainty. These take the form of singular control, which is ubiquitous in optimal control problems under uncertainty with linear cost and means that one cannot quantify the rate at which optimal investment happens. Instead, one needs to characterize the cumulative investment (or control in general), and doing this with feedback strategies is conceptually much more difficult in a game than for a single agent. Nevertheless, the concept developed here permits equilibrium strategies that have an intuitive representation and are easy to interpret.

The paper is organized as follows. Section~\ref{sec:model} presents the basic model. The novel concept of closed-loop equilibrium is developed in Section~\ref{sec:eql}. Section~\ref{sec:approach} describes the solution approach and develops a central verification theorem. Since the approach relies on explicit solutions, Section~\ref{sec:application} specifies the model more concretely for the subsequent results. Section~\ref{sec:static} proves that the zero-NPV rule is supported by a closed-loop equilibrium, while Section~\ref{sec:dynamic} proves that more reactive strategies support also positive option values in closed-loop equilibria. These results are further discussed in Section~\ref{sec:discussion}. Finally, three appendices contain most of the details for the formal proofs.

%\newpage

\section{Model}\label{sec:model}

Two firms accumulate capital by irreversible investment under uncertainty. Time $t\geq 0$ is continuous and there is a given probability space $(\Omega,\F,P)$ with a filtration $(\F_t)$ that represents the dynamic information about the state of the world. Denote the capital stock of firm $i$ at time $t$ by $Q_t^{i}$. It is assumed that foresight is impossible, investment is irreversible, but new capital is installed instantaneously. Thus, a \emph{capital stock process} $Q^{i}=(Q_t^{i})$ is adapted to the given filtration, nondecreasing, and right-continuous. 

The cost of an additional unit of capital is normalized to one at any time, and there is a constant required rate of return $r>0$. Thus, the expected net present value of the investment cost for $Q^{i}$ is $E\int_0^{\infty}e^{-rt}\,dQ_t^{i}$. Given some initial capital stock $q^{i}\geq 0$, a capital stock process is \emph{admissible} if the investment cost is finite and $Q_0^{i}\geq q^{i}$. Denote the set of admissible capital stock processes by $\A(q^{i})$.

Each firm's operating profit flow depends on the capital stocks of both firms. Additionally, there is an exogenous economic shock $X=(X_t)$ that solves a stochastic differential equation
\begin{equation*}
dX_t=\mu(X_t)\,dt+\sigma(X_t)\,dB_t,
\end{equation*}
where $(B_t)$ is a one-dimensional Brownian motion on the filtered probability space. Assume the functions $\mu(x)$ and $\sigma(x)$ satisfy the standard conditions for existence of a unique strong solution for any initial condition $X_0=x$.\footnote{That is, there exist constants $c$ and $d$ such that $\abs{\mu(x)}+\abs{\sigma(x)}\leq c(1+\abs{x})$ and $\abs{\mu(x)-\mu(y)}+\abs{\sigma(x)-\sigma(y)}\leq d\abs{x-y}$ for all $x$ and $y$.}
Referring to the respective other firm by $-i$, the operating profit of firm $i$ at time $t$ is $\pi(X_t,Q_t^{i},Q_t^{-i})$. Assume the function $\pi$ is bounded from below.\footnote{Alternatively, one can assume that $e^{-rt}\pi(X_t,Q_t^{i},Q_t^{-i})$ is $P\otimes dt$-integrable for admissible $Q^{i}$ and $Q^{-i}$.}
Then firm $i$'s \emph{payoff}
\begin{equation*}
\Pi(Q^{i},Q^{-i})=E\int_0^{\infty}e^{-rt}\pi(X_t,Q_t^{i},Q_t^{-i})\,dt-E\int_0^{\infty}e^{-rt}\,dQ_t^{i}
\end{equation*}
is well defined for any pair of admissible capital stock processes.\footnote{At this point, $\Pi(Q^{i},Q^{-i})$ potentially takes the value $\infty$, but appropriate assumptions on the parameters in the considered applications will ensure finite equilibrium payoffs.}

Note that, in order to focus on the pure strategic effect of capital commitment, the firms' payoffs differ only by the attribution of the capital stock processes.

%We want to formulate a stochastic continuous-time model, in which two players strategically accumulate capital by irreversible investment. Their respective objective is to maximize the value of a profit flow depending on both capital levels and exogenous uncertainty, net of investment costs. Formally, when the capital stock processes of player $i$ and the opponent are $Q^{i}$ and $Q^{-i}$, the payoff of player $i$ is given by
%%
%\begin{equation}\label{J_CL}
%\Pi(Q^{i},Q^{-i})\triangleq\Exp\left[\int_0^{\infty} e^{-rt}\pi(X_t,Q^{i}_t,Q^{-i}_t)\,dt-\int_0^{\infty} e^{-rt}\,dQ^{i}_t\right],
%\end{equation}
%%
%with a constant discount rate $r>0$. Since we focus on the pure strategic effect of capital commitment, the payoffs to the players differ only through the capital stock processes. The \emph{instantaneous revenue} function $\pi$ is further affected by an exogenous stochastic process $X$, which is defined on the probability space $(\Omega,\filt{ },\Pb)$ and adapted to the filtration $(\filt{t})_{t\geq0}$. Assume the latter satisfies the usual conditions of right-continuity and completeness. 

%\newpage

\section{Equilibrium concept}\label{sec:eql}

If a firm were to directly choose a capital stock process, it would use an open-loop strategy. It would condition investment on the evolution of the exogenous shock but not react to the actual capital stock of the other firm. \cite{BackPaulsen09} showed that equilibria with such strategies, like in \cite{Grenadier02}, are not subgame perfect---the lack of reactions is not credible. But \cite{BackPaulsen09} also argued that formulating closed-loop strategies that allow reactions to actual behavior raises fundamental conceptual difficulties. The novel concept developed in this section overcomes these issues.

The first issue is that optimal investment in all the related models of monopoly, oligopoly with open-loop strategies, and perfect competition takes the form of singular control. This means that $Q_t^{i}$ increases at a rate $dQ_t^{i}$ that is zero almost everywhere and otherwise undefined. Therefore, it is impossible to quantify the capital increments or investments. Instead, strategies need to relate to the levels of the capital stocks or cumulative investments. However, as capital stocks must be nondecreasing, it is not viable either to specify $Q_t^{i}$ by a function of some current state.\footnote{Besides, choosing actions in continuous time would raise the issues discussed by \cite{SimonStinchcombe89}.} 

The solution proposed here is inspired by a general characterization of optimal irreversible investment that holds in the cases of monopoly and open-loop equilibrium.\footnote{See, respectively, \cite{RiedelSu11} and \cite{Steg12}.}
It involves a ``base capacity'' process such that each firm just invests enough to keep its capital stock above that base capacity. Here, the base capacity will be a function $\phi^{i}$ of the state relevant for firm $i$, which consists of the current level of the exogenous shock and the other firm's capital stock. This means
\begin{equation}\label{phi_i}
Q_t^{i}=q^{i}\vee\sup_{0\leq s\leq t}\phi^{i}(X_s,Q_s^{-i}).
\end{equation}
$\phi^{i}$ is firm $i$'s strategy.

However, these strategies %, which will take the form of investment triggers, 
are still subject to the second issue discussed in \cite{BackPaulsen09}: intuitively appealing strategies often fail to imply a unique outcome. This is in fact a general issue for continuous-time games.\footnote{This issue was already discussed by \cite{Anderson84} and \cite{SimonStinchcombe89}.}
Here, an \emph{outcome} for a given initial state $(x,q^{1},q^{2})$ %, which implicates $X_0=x$, 
is an admissible pair $(Q^{1},Q^{2})\in\cA(q^{1})\times\cA(q^{2})$. Requiring the system of equations consisting of \eqref{phi_i} for each $i$ to define a unique outcome would rule out even simple trigger strategies.

Instead, the solution proposed here is to require an outcome that is optimal for both firms in the following strong sense: No firm $i$ can find any better outcome that could occur given the other firm's strategy $\phi^{-i}$. This guarantees each firm $i$ an optimal outcome. In particular, there is no need to worry that the outcome might be ambiguous by using any strategy $\phi^{i}$ of the given form.  

Formally, an outcome $(Q^{1},Q^{2})$ is said to be \emph{consistent with} a strategy $\phi^{-i}$ for the other firm if equation \eqref{phi_i} holds for that firm, which means
\begin{equation*}%\label{phi_-i}
Q_t^{-i}=q^{-i}\vee\sup_{0\leq s\leq t}\phi^{-i}(X_s,Q_s^{i}).
\end{equation*}
These are the outcomes that firm $i$ is allowed to consider. An equilibrium outcome must be consistent with both firms' strategies simultaneously. This means it needs to satisfy the system of equations 
\begin{equation*}%\label{outcome_from}
\left.\begin{aligned}
Q_t^{1}&=q^{1}\vee\sup_{0\leq s\leq t}\phi^{1}(X_s,Q_s^{2}), \\
Q_t^{2}&=q^{2}\vee\sup_{0\leq s\leq t}\phi^{2}(X_s,Q_s^{1}).
\end{aligned}\;\;\right\}
\end{equation*}
Any such outcome is said to be an \emph{outcome from} the strategy pair $(\phi^{1},\phi^{2})$.

Given that the exogenous shock $X$ is a time-homogeneous Markov process and that the past evolution of capital stocks does not matter for future profits, the equilibrium condition is required to hold for any relevant state without explicit consideration of time.

\begin{definition}\label{MPE}
A pair of strategies $(\phi^{1},\phi^{2})$ is a \emph{Markov perfect equilibrium} for initial capital stocks $(q_0^{1},q_0^{2})$ if for every state $(x,q^1,q^2)$ with $q^{1}\geq q_0^{1}$ and $q^{2}\geq q_0^{2}$ there exists an outcome from $(\phi^{1},\phi^{2})$ that, for each $i$, maximizes $\Pi(Q^{i},Q^{-i})$ among all outcomes that are consistent with $\phi^{-i}$.
\end{definition}

%\newpage

\section{Solution approach}\label{sec:approach}

A Markov perfect equilibrium requires to solve the optimization problems
\begin{equation*}%\label{opt_i}
\begin{aligned}
&\max\; \Pi(Q^{i},Q^{-i}),\quad &&(Q^{i},Q^{-i})\in\A(q^{i})\times\A(q^{-i}), \\[6pt]
&\text{s.t.} &&Q_t^{-i}=q^{-i}\vee\sup_{0\leq s\leq t}\phi^{-i}(X_s,Q_s^{i}).
\end{aligned}
\end{equation*}
In the case of open-loop strategies, firm $i$ takes some $Q^{-i}\in\A(q^{-i})$ as given and chooses only $Q^{i}\in\A(q^{i})$. Then it turns out that optimal investment is \emph{myopic}: ignoring all future investments by any firm, just determine the optimal time to add a marginal unit.\footnote{Formally this means that the connection between singular control problems of the monotone follower type and optimal stopping problems established by \cite{KaratzasShreve84} still holds in open-loop equilibrium.} 
Such a reduction to optimal stopping is no longer possible with closed-loop strategies, because they imply path dependence: Now the ``continuation payoff'' at any future time $t$ is not only affected by $Q_t^{i}$ but, via $Q_t^{-i}$, by the whole previous path $(Q_s^{i})_{0\leq s\leq t}$. Therefore, a dynamic programming approach will be followed here, based on a central verification theorem that takes into account the other firm's strategy $\phi^{-i}$.

\subsection{Reflection strategies}\label{sec:reflection}

In order to develop the verification theorem, suppose the firms invest when the value of the economic shock is sufficiently high, and where the critical value depends on the current capital stocks. This means firm $i$ invests in any state such that $x>\bar X^{i}(q^{i},q^{-i})$. The function $\bar X^{i}$ is firm $i$'s investment trigger. Assume it is continuous, strictly increasing in $q^{i}$, and such that $\bar X^{i}\to\infty$ as $q^{i}\to\infty$. This implies that the investment trigger rises above $x$ whenever firm $i$ invests enough. Further, to prevent investment circles, assume $\bar X^{i}$ is nondecreasing in $q^{-i}$. Given these properties, it is possible to define a strategy by the minimal capital level such that no further investment is triggered, i.e.,
\begin{equation}\label{phiX}
\phi^{i}(x,q^{-i})=\inf\{q\geq 0\mid x\leq\bar X^{i}(q,q^{-i})\}.
\end{equation}
Then indeed
\begin{equation}\label{q<phi}
q^{i}<\phi^{i}(x,q^{-i}) \iff x>\bar X^{i}(q^{i},q^{-i})
\end{equation}
and
\begin{equation}\label{q=phi}
q^{i}=\phi^{i}(x,q^{-i}) \iff x=\bar X^{i}(q^{i},q^{-i}).
\end{equation}
%\footnote{By definition, $x>\bar X^{i}$ whenever $q^{i}<\phi^{i}$, and by monotonicity and continuity of $\bar X^{i}$, $x\leq\bar X^{i}$ whenever $q^{i}\geq\phi^{i}$. Together, this yields \eqref{q<phi}. Moreover, by strict monotonicity of $\bar X^{i}$ in $q^{i}$, $x<\bar X^{i}$ whenever $q^{i}>\phi^{i}$, and by continuity, $x\geq\bar X^{i}$ whenever $q^{i}\leq\phi^{i}$, which yields that $q^{i}>\phi^{i}$ if and only if $x<\bar X^{i}$. This equivalence and \eqref{q<phi} together yield \eqref{q=phi}.}
%, whereas $q^{i}\geq\phi^{i}(x,q^{-i})$ for all states if $\bar X^{i}=\infty$.
Alternatively, firm $i$ may also consider a ``trigger'' $\bar X^{i}$ identically equal to $\infty$. Then equation \eqref{phiX} implies that $\phi^{i}$ is identically equal to zero, so that firm $i$ never invests and equivalence \eqref{q<phi} trivially remains true.

Any strategy of the form \eqref{phiX} keeps the state in the region $\{x\leq\bar X^{i}(q^{i},q^{-i})\}$ with minimal effort. After a single discrete investment if initially $x>\bar X^{i}(q^{i},q^{-i})$, the state is reflected at the boundary $\bar X^{i}$ by singular control. Hence, any such $\phi^{i}$ is said to be a \emph{reflection strategy} if the corresponding boundary $\bar X^{i}$ has the assumed properties and, for technical reasons, is continuously differentiable in any $q^{i}>0$ (if $\bar X^{i}$ is finite).

\subsection{Verification theorem}\label{sec:ver}

The following Theorem~\ref{thm:ver} is the main tool for verifying Markov perfect equilibria with reflection strategies. As usual, it provides sufficient conditions for a function $V$ defined on the state space to yield the payoff from a given strategy and also to dominate the payoff from other strategies. But it differs from other verification theorems by the presence of the other firm's strategy $\phi^{-i}$ and further by the fact that firm $i$ can choose its preferred outcome.

The starting point is the typical differential equation that $V$ should solve in the region in which the firms do not invest, because then the payoff evolves like an asset whose price is a function of $X$ alone and that generates a dividend flow $\pi$ (condition~\ref{V_PDE}). In firm $i$'s investment region, the function $V$ anticipates that firm $i$ invests: the change in value just offsets the investment cost, which is one per unit of capital (condition~\ref{V_qi=1}). In contrast, it is not expected that the other firm necessarily invests in its investment region; this will depend on the outcome chosen by firm $i$. It is only anticipated that the other firm invests when its investment region does not intersect that of firm $i$, and then this does not affect firm $i$'s value (condition~\ref{V_q-i=0}).
 
The investments that firm $i$ ``allows'' the other firm to make in the joint investment region in the chosen outcome are anticipated if the implications \eqref{V_q-i_0} and \eqref{V_q-i_t} hold, where the former concerns a discrete initial investment and the latter reflection investments. In both cases, it is enough to consider states on the boundary $\bar X^{i}$ as the theorem shows.

Finally, $V$ needs to satisfy some technical integrability conditions when applied to an outcome, and the further sufficient conditions imply that firm $i$ cannot gain from choosing a different outcome.

\begin{theorem}[Verification]\label{thm:ver}
Let $\phi^{1}$ and $\phi^{2}$ be reflection strategies with corresponding boundaries $\bar X^{1}$ and $\bar X^{2}$. Suppose there exists a differentiable function $V(x,q^{i},q^{-i})$ such that
\begin{enumerate}[label=\arabic*.]
\item\label{V_PDE} on $\{x\leq\bar X^{1}(q^{1},q^{2})\wedge\bar X^{2}(q^{2},q^{1})\}$, $V$ is continuously differentiable, twice with respect to $x$, and satisfies the differential equation
\begin{equation*}%\label{PDE}
-rV+\pi+\mu V_x+\frac12{\sigma}^2V_{xx}=0,
\end{equation*}
\item\label{V_qi=1} on $\{x\geq\bar X^{i}(q^{i},q^{-i})\}$, $V_{q^{i}}=1$, %and $V$ is differentiable with respect to $(q^{i},q^{-i})$ for every $x=\bar X^{i}(q^{i},q^{-i})$
\item\label{V_q-i=0} on $\{\bar X^{i}(q^{i}q^{-i})>\bar X^{-i}(q^{-i},q^{i})\}\cap\{\bar X^{i}(q^{i},q^{-i})\geq x\geq\bar X^{-i}(q^{-i},q^{i})\}$, $V_{q^{-i}}=0$. %and $V$ is differentiable with respect to $(q^{i},q^{-i})$ for every $x=\bar X^{-i}(q^{-i},q^{i})<\bar X^{i}(q^{i},q^{-i})$
\end{enumerate}%\setcounter{enum+}{\value{enumi}}
Then, for any given state $(x,q^{1},q^{2})$ and $X_0=x$,
\begin{equation*}
V(x,q^{i},q^{-i})=\Pi(Q^{i},Q^{-i})
\end{equation*}
for any outcome from $(\phi^{1},\phi^{2})$ that satisfies
%\begin{enumerate}[resume*]%\addtocounter{enumi}{\value{enum+}}
%\item
\begin{equation}\label{V_q-i_0}
\forall q\in(q^{-i},Q_0^{-i}):\quad Q_0^{i}<\phi^{i}(x,q)\implies V_{q^{-i}}(x,\phi^{i}(x,q),q)=0,
%x_0>\bar X^{i}(Q^{i},q^{-i})
\end{equation}
\begin{equation}\label{V_q-i_t}
\forall t>0:\quad \indi{X_t=\bar X^{i}(Q_t^{i},Q_t^{-i})=\bar X^{-i}(Q_t^{-i},Q_t^{i})}\,dQ_t^{-i}>0\implies V_{q^{-i}}(X_t,Q_t^{i},Q_t^{-i})=0,
%x_0>\bar X^{i}(Q^{i},q^{-i})
\end{equation}
\begin{equation}\label{Vintegrable}
\forall T>0:\quad E\left[\sup_{t\leq T}\left\lvert V(X_t,Q_t^{i},Q_t^{-i})\right\rvert\right]<\infty,
%x_0>\bar X^{i}(Q^{i},q^{-i})
\end{equation}
and
%\item
\begin{equation}\label{limV=0}
\lim_{T\to\infty}E\left[e^{-rT}V(X_T,Q_T^{i},Q_T^{-i})\right]=0.
\end{equation}
%\end{enumerate}%\setcounter{enum+}{\value{enumi}}
Furthermore, if $V$ has the additional properties that
\begin{enumerate}[resume*]%\addtocounter{enumi}{\value{enum+}}
\item\label{V_PDE_leq} on $\{x\leq\bar X^{-i}(q^{-i},q^{i})\}$, $V$ is continuously differentiable, twice with respect to $x$, and
\begin{equation*}%\label{PDE_leq}
-rV+\pi+\mu V_x+\frac12{\sigma}^2V_{xx}\leq 0,
\end{equation*}
\item\label{V_qi<1} on $\{x\leq\bar X^{1}(q^{1},q^{2})\wedge\bar X^{2}(q^{2},q^{1})\}$, $V_{q^{i}}\leq 1$,
\item\label{V_q-i<0} on $\{x=\bar X^{i}(q^{i},q^{-i})\leq\bar X^{-i}(q^{-i},q^{i})\}$, $V_{q^{-i}}\leq 0$,
\end{enumerate}
then, for any given state $(x,q^{1},q^{2})$ and $X_0=x$,
\begin{equation*}
V(x,q^{i},q^{-i})\geq\Pi(Q^{i},Q^{-i})
\end{equation*}
for any outcome that is consistent with $\phi^{-i}$ and satisfies \eqref{Vintegrable} and \eqref{limV=0}.
%\begin{equation}\label{limV=0}
%\limsup_{T\to\infty}E\left[e^{-rT}V(X_T,Q_T^{i},Q_T^{-i})\right]\geq 0.
%\end{equation}
\end{theorem}

The proof of Theorem~\ref{thm:ver} is given in Appendix~\ref{app:ver}.

%\newpage

\section{Application}\label{sec:application}

As the approach taken here relies on explicit solutions, consider from now on the following application from \cite{Grenadier02}. The firms produce a homogeneous good at full capacity and sell it on a common market. Inverse demand has constant elasticity and is multiplicatively affected by the exogenous shock. With zero variable cost, the operating profit for firm $i$ then is
\begin{equation*}%\label{xPq}
\pi(x,q^{i},q^{-i})=xP(q^{i}+q^{-i})q^{i}=x(q^{i}+q^{-i})^{-\frac{1}{\gamma}}q^{i}.
\end{equation*}
Assume $\gamma>1$, which implies that marginal profit with respect to $q^{i}$ is positive, increasing in $x$, but decreasing in $q^{i}$. 

Further, the shock process $X$ is now a geometric Brownian motion, so it solves
\begin{equation*}
dX_t=\mu X_t\,dt+\sigma X_t\,dB_t
\end{equation*}
for some constants $\mu$ and $\sigma\neq 0$, and the state space for $X$ is $\{x>0\}$. In order to ensure finite equilibrium values, assume in addition to $r>0$ that
\begin{equation}\label{r>mu_gamma}
r>\gamma\mu+\gamma(\gamma-1)\frac12\sigma^2.
\end{equation} 
Since $\gamma>1$, this can also be regarded as a restriction on the growth rate $\mu$ or the volatility $\sigma$. It implies that $r>\mu$, and it is in fact equivalent to $\beta>\gamma$, where $\beta$ is the unique positive number that, in place of $\gamma$, turns \eqref{r>mu_gamma} into an equality. This $\beta$ is well known from (real) option pricing models and will also play an important role in the value functions here.

%\newpage

\section{Static price triggers}\label{sec:static}

Since investment is fully divisible and precommitment not credible, and the only externality is that inverse demand decreases, it is natural to expect that the firms will preempt any profitable investments. This would mean that in equilibrium they make no profit and effectively follow the simple zero-NPV rule. Thus, competition would completely eliminate the option value of waiting that is still present in open-loop equilibrium. 

The aim of this section is to formally prove that this reasoning is correct if the firms use a very simple type of reflection strategy: to invest whenever inverse demand $xP(q^{i}+q^{-i})$ rises above a constant threshold $p>0$. 

%Hence, assume the other firm uses such the corresponding reflection boundary
%\begin{equation*}
%\bar X^{-i}(q^{-i},q^{i})=\frac{p}{P(q^{i}+q^{-i})}=p(q^{i}+q^{-i})^{\frac{1}{\gamma}}.
%\end{equation*}
%Then firm $i$ can let the other firm's investments happen or make two types of own investment. The first type is to invest (at most) as much as the other firm would at the same instant to keep inverse demand below $p$. This replaces (part of) the other firm's investment and, thus, the decrease of inverse demand is no externality on existing output; it would happen in any case. The second type is such that inverse demand decreases more than it otherwise would; then the externality $xP'q^{i}$ on existing output becomes effective. Nevertheless, this can still be profitable.
%
%In fact, it is possible to determine a best reply for firm $i$ using the Verification Theorem~\ref{thm:ver}, but to prove the result, it is enough to consider only two simple options. 

Such a rule is known from perfect competition between firms that correspond to marginal units of capital; see \cite{Leahy93}. In a perfectly competitive equilibrium, the net present value of a marginal investment at the threshold must be zero. This determines the equilibrium threshold
\begin{equation*}
p^{*}=(r-\mu)\frac{\beta}{\beta-1}.
\end{equation*}

The basic idea why this constant threshold $p^{*}$ also yields a Markov perfect equilibrium in the present setting is the following. No firm can gain by investing, because the other firm's strategy implies that inverse demand never exceeds $p^{*}$; so the net present value is at best zero. But each firm is also indifferent to make the necessary investments, because then the other firm will not invest; so inverse demand is indeed on the boundary when investment happens, and then the net present value is exactly zero. This reasoning can now be formalized with the equilibrium concept developed in Section~\ref{sec:eql}, by showing that the given strategies indeed admit suitable outcomes.

At any threshold $p>p^{*}$, however, investments have positive net present value. Then it is clearly more profitable for each firm to make all necessary investments than to abstain. Thus, in equilibrium, both firms would have to make some investments. But it turns out that also such outcomes cannot be optimal for both firms.

\begin{theorem}\label{thm:static}
There exists a unique Markov perfect equilibrium with a constant price threshold, which is $p=p^{*}$. Any investment in this equilibrium has zero net present value.
\end{theorem}

\begin{proof}%[Proof of Theorem~\ref{thm:static}]
The main arguments are given here, and details are worked out in a sequence of lemmas in Appendix~\ref{app:static}.
%Take the point of view of any firm $i$ and assume the other firm uses a reflection strategy $\phi^{-i}$ with boundary
%\begin{equation*}
%\bar X^{-i}(q^{-i},q^{i})=\frac{p}{P(q^{i}+q^{-i})}=p(q^{i}+q^{-i})^{\frac{1}{\gamma}}
%\end{equation*}
%for some $p>0$.
Let $\phi^{1}$ and $\phi^{2}$ be reflection strategies that correspond to a constant price threshold $p>0$. Then Lemma~\ref{lem:p_admissible} shows that, for each $i$, there exists an admissible outcome from $(\phi^{1},\phi^{2})$ such that firm $i$ does not invest. In particular, the investment cost falling to the other firm is finite by \eqref{r>mu_gamma}. Using the Verification Theorem~\ref{thm:ver}, Lemma~\ref{lem:Vinfty} shows that the value of firm $i$'s payoff from abstaining is given by the function
%\begin{equation}\label{Vinfty}
%V^{\text{abs}}(x,q^{i},q^{-i})=\begin{cases} A(q^{i},q^{-i})x+B^{\text{abs}}(q^{i},q^{-i})x^{\beta} &\text{if }x\leq\bar X^{-i}(q^{-i},q^{i}) \\ V^{\text{abs}}(x,q^{i},\phi^{-i}(x,q^{i})) &\text{else,} \end{cases}
%\end{equation}
%where 
%\begin{equation*}
%B^{\text{abs}}(q^{i},q^{-i})=-\frac{p}{\beta(r-\mu)}q^{i}\left(\frac{P(q^{i}+q^{-i})}{p}\right)^{\beta}
%\end{equation*}
\begin{equation}\label{Vinfty}
V^{\text{abs}}(x,q^{i},q^{-i})=\begin{cases} \dfrac{p}{r-\mu}\left(\dfrac{xP(q^{i},q^{-i})}{p}-\dfrac{1}{\beta}\left(\dfrac{xP(q^{i},q^{-i})}{p}\right)^{\beta}\right)q^{i}
 &\text{if }xP(q^{i},q^{-i})\leq p, \\[16pt] \dfrac{p}{r-\mu}\dfrac{\beta-1}{\beta}q^{i} &\text{if }xP(q^{i},q^{-i})>p. \end{cases}
\end{equation}
It solves the differential equation in condition~\ref{V_PDE}, because it is of the form $Ax+Bx^{\beta}$ for $x\leq\bar X^{-i}=p/P$, where $Ax=\pi/(r-\mu)$. The other firm's investments are anticipated, since by construction $V_{q^{-i}}^{\text{abs}}=0$ for $x\geq \bar X^{-i}=p/P$. This is easy to verify also on the boundary, where $xP/p=1$. Further, if $p\leq p^*$, then this outcome is optimal among all that are consistent with $\phi^{-i}$. 

Again using the Verification Theorem~\ref{thm:ver}, Lemma~\ref{lem:Vp} shows that the value of firm $i$'s payoff from the converse outcome, such that only firm $i$ invests, is given by the function
\begin{flalign}\label{Vp}
&& V^{\text{inv}}(x,q^{i},q^{-i})&=\begin{cases} 
\dfrac{xP(q^{i},q^{-i})q^{i}}{r-\mu}+B(q^{i},q^{-i})x^{\beta}
 &\text{if }xP(q^{i},q^{-i})\leq p, \\[10pt] 
 V^{\text{inv}}(x,\phi^{i}(x,q^{-i}),q^{-i})-\phi^{i}(x,q^{-i})+q^{i} &\text{if }xP(q^{i},q^{-i})>p, 
\end{cases} && \\[10pt]
%&& & && \\
&\text{where}& B(q^{i},q^{-i})x^{\beta}&=\dfrac{\gamma}{\beta-\gamma}\left(\left(\dfrac{p(\gamma-1)}{(r-\mu)\gamma}-1\right)q^{i}+\left(\dfrac{p(\beta-1)}{(r-\mu)\beta}-1\right)q^{-i}\right)\left(\dfrac{xP(q^{i},q^{-i})}{p}\right)^{\beta}. &\hphantom{\text{where}}& \nonumber
\end{flalign}
$V^{\text{inv}}$ solves the differential equation in condition~\ref{V_PDE} for the same reason as $V^{\text{abs}}$, but to anticipate that firm $i$ invests, the function $B(q^{i},q^{-i})$ is now constructed such that $V_{q^{i}}^{\text{inv}}=1$ on the boundary $x=\bar X^{i}=p/P$. This is again easy to verify by $xP/p=1$, and clearly also $V_{q^{i}}^{\text{inv}}=1$ for $x>\bar X^{i}=p/P$.
%\footnote{It is also possible to give the following simplified expression for $V^{\text{inv}}$ in case $xP>p$, but this seems less instructive than the one for $V^{\text{abs}}$:
%\begin{equation*}
%V^{\text{inv}}=\left(\frac{p(\beta-1)}{(r-\mu)\beta}-1\right)\frac{\beta\phi^{i}+\gamma q^{-i}}{\beta-\gamma}+q^{i}\quad\text{if }xP>p.
%\end{equation*}}

Lemma~\ref{lem:Vp>Vinfty} implies that $V^{\text{inv}}=V^{\text{abs}}$ for $p=p^{*}$, so that firm $i$ is indifferent who abstains. Since abstaining was shown to be optimal for $p\leq p^{*}$ and $i$ was arbitrary, it follows that both outcomes are optimal for both firms if $p=p^{*}$. Hence, either one supports a Markov perfect equilibrium. 

A threshold $p>p^{*}$ cannot be an equilibrium, because then $V^{\text{inv}}>V^{\text{abs}}$ by Lemma~\ref{lem:Vp>Vinfty}, whereas Lemma~\ref{lem:p_eql_payoff} shows that the payoffs in any equilibrium with a constant price threshold must be equal to the value of $V^{\text{abs}}$.
%\footnote{It is also possible to determine a best reply to the boundary $\bar X^{-i}=p/P$ with the help of the Verification Theorem~\ref{thm:ver}, but to prove the result, is it enough to find any profitable deviation from not investing at all.}

A threshold $p<p^{*}$ cannot be an equilibrium, because then the return to any investment will be strictly negative, but there is no outcome from $(\phi^{1},\phi^{2})$ such that both firms abstain.\footnote{This follows formally from the fact that $V_{q^{i}}^{\text{abs}}<1$ for any $p<p^{*}$ by Lemma~\ref{lem:Vinfty_qi}, so that the inequality $V(x,q^{i},q^{-i})\geq\Pi(Q^{i},Q^{-i})$ in the Verification Theorem~\ref{thm:ver} will be strict for any $Q^{i}$ that ever increases with positive probability.}

Finally, note that the net present value of any investment is at most equal to $(V^{\text{abs}}/q^{i})-1$ per unit. Indeed, the latter is the net present value of an investment that does not affect inverse demand, which means that it only replaces an investment by the other firm. It is also an upper bound for the net present value of any other investment, because any effect on inverse demand is always negative. Further, $V^{\text{abs}}\leq (p/p^{*})q^{i}$ by Lemma~\ref{lem:Vinfty_qi}, where equality holds if and only if $x\geq\bar X^{-i}$. It follows that the net present value of any optimal investment is exactly zero if $p=p^{*}$. 
\end{proof}

The outcomes chosen in the proof are such that one firm abstains and only the other invests. But Lemmas \ref{lem:reflection_sum} and \ref{lem:indifference} in Appendix~\ref{app:static} imply that it is also possible to choose any other outcome that generates the same aggregate capital $Q^{1}+Q^{2}$.
%\footnote{By Lemma~\ref{lem:reflection_sum}, any such outcome is consistent with the two reflection strategies. Therefore, among all these outcomes, the equilibrium outcome must be optimal for each firm. By Lemma~\ref{lem:indifference} then, however, both firms must in fact be indifferent between all these outcomes, so that indeed each of them is an equilibrium outcome.}
This is a direct consequence of the fact that the reflection boundaries $\bar X^{i}=\bar X^{-i}$ depend only on the sum $q^{i}+q^{-i}$ (as the proofs of these lemmas reveal).

%\newpage

\section{Dynamic price triggers}\label{sec:dynamic}

This section is going to develop the main result of the paper. Although each firm has the opportunity to perfectly preempt the other firm's investment plan, it is possible to sustain positive profits in equilibrium. This requires strategies that depend on the individual capital stocks, and not only on the sum, but the representation will still be a similarly simple.

With a constant price threshold $p$, the preemption decision is static. Preemptive investments have no consequences for future investment opportunities, because they do not affect the dynamics of inverse demand. Therefore, to break preemption incentives, a firm must face an investment trigger that depends on its own capital---preemptive investments need to have an adverse effect on future profits.

The following investment triggers result from the consideration that in order to prevent excessive capital accumulation through preemption, the bigger firm should be indifferent to invest. This is not required for the smaller firm, who may well have a strict preference for investment.\footnote{It is possible to show that if one requires a firm to be indifferent at the entire boundary $\bar X^{-i}$, then $\bar X^{-i}$ must be the constant price boundary $p^{*}/P$ from Theorem~\ref{thm:static}.} 

Fix any $c\geq 0$ and suppose the firms use reflection strategies with boundaries
\begin{equation}\label{barXc}
\bar X^{i}(q^{i},q^{-i})=\bar X^{-i}(q^{-i},q^{i})=\left(p^{*}+\frac{c}{q^{i}\vee q^{-i}}\right)\left(P(q^{i}+q^{-i})\right)^{-1},\qquad q^{i},q^{-i}\geq c\frac{2\gamma-1}{p^{*}},
\end{equation}
where the lower bound on capital stocks ensures that the boundaries are indeed strictly increasing in $q^{i}$ and $q^{-i}$ (see Lemma~\ref{lem:barXc} in Appendix~\ref{app:dynamic}).

For any $c>0$, these boundaries imply a higher investment threshold for inverse demand than the zero-NPV threshold $p^{*}$. This means that the corresponding investments have strictly positive returns. But investments by the bigger firm decrease the threshold at which future investments will happen and, accordingly, decrease future inverse demand. This impact deters from preemption and permits that, in equilibrium, first only the smaller firm invests and then both firms grow symmetrically.

\begin{theorem}\label{thm:dynamic}
%There exists a unique MPE with a constant price threshold. This threshold is
%\begin{equation*}
%p^{*}=(r-\mu)\frac{\beta}{\beta-1},
%\end{equation*}
%and any investment in this equilibrium has zero net present value.
The pair of reflection strategies with boundaries given by \eqref{barXc} is a Markov perfect equilibrium for any $c\geq 0$ and initial capital stocks $q_0^{1},q_0^{2}\geq c(2\gamma-1)/p^{*}$.
\end{theorem}

\begin{proof}%[Proof of Theorem~\ref{thm:dynamic}]
Let $(\phi^{1},\phi^{2})$ be a pair of reflection strategies with boundaries given by \eqref{barXc} for some $c\geq 0$, and consider any state such that $q^{1},q^{2}\geq c(2\gamma-1)/p^{*}$. The proof is going to determine an outcome from $(\phi^{1},\phi^{2})$ that is optimal for both firms. The main steps are outlined here, and the details are worked out in the referenced lemmas in Appendix~\ref{app:dynamic}. 

The outcome is formally defined by
\begin{equation}\label{c_eql_outcome}
Q_t^{i}=q^{i}\vee\sup_{0\leq s\leq t}\left(\phi^{i}(X_s,q^{-i})\wedge \psi^{c}(X_s)\right),
\end{equation}
where
\begin{equation*}%\label{psi}
\psi^{c}(x)=\inf\{q\geq c(2\gamma-1)/p^{*}\mid\bar X^{i}(q,q)\geq x\}
\end{equation*}
is the minimal capital level that symmetric firms need to have to keep $\bar X^{i}$ above $x$. Lemma~\ref{lem:c_eql_outcome} proves that equation \eqref{c_eql_outcome} indeed defines an admissible outcome, and---most importantly---that this outcome is consistent with the given reflection strategies.
%In particular the latter requires more work than it was the case for the simple outcomes used in the proof of Theorem~\ref{thm:static}.
Along the proof, it is also shown that this outcome has the described property that first only the smaller firm invests until it catches up to the bigger firm, and from then on both firms grow symmetrically.

In order to use the Verification Theorem~\ref{thm:ver} to prove optimality of this outcome for each firm $i$, the candidate value function is constructed by
\begin{flalign}\label{Vc}
&& V^{c}(x,q^{i},q^{-i})&=\begin{cases} 
\dfrac{xP(q^{i},q^{-i})q^{i}}{r-\mu}+B(q^{i},q^{-i})x^{\beta}
 &\text{if }x\leq\bar X^{-i}(q^{-i},q^{i}), \\[10pt] 
 V^{c}(x,\phi^{i}(x,q^{-i}),q^{-i})-\phi^{i}(x,q^{-i})+q^{i} &\text{if }x>\bar X^{-i}(q^{-i},q^{i}), 
\end{cases} && \\[10pt]
%&& & && \\
&\text{where}& B(q^{i},q^{-i})&=-\int_{q^{i}}^{\infty}\left(1-\bar X^{-i}(q,q^{-i})\dfrac{P'(q,q^{-i})q+P(q,q^{-i})}{r-\mu}\right)\left(\bar X^{-i}(q,q^{-i})\right)^{-\beta}\,dq. &\hphantom{\text{where}}& \nonumber
\end{flalign}
Lemma~\ref{lem:Vc_q-i} verifies that this function is well defined by the given integral. $V^{c}$ solves the differential equation in condition~\ref{V_PDE}, because it has the same general form for $x\leq\bar X^{-i}$ as $V^{\text{inv}}$ in the proof of Theorem~\ref{thm:static}. And like $V^{\text{inv}}$, it anticipates investments by firm $i$ given that $\bar X^{i}=\bar X^{-i}$, because $B(q^{i},q^{-i})$ is still constructed such that $V_{q^{i}}^{c}=1$ when $x=\bar X^{-i}$, and clearly $V_{q^{i}}^{c}=1$ for $x>\bar X^{-i}$. To anticipate additionally the other firm's investments, however, $\bar X^{-i}$ is now constructed such that $V_{q^{-i}}^{c}=0$ when $x=\bar X^{-i}$ and $q^{i}\geq q^{-i}$, which as well is verified in Lemma~\ref{lem:Vc_q-i}.

Using these properties and the Verification Theorem~\ref{thm:ver}, Lemma~\ref{lem:Vc} then proves that the outcome given by \eqref{c_eql_outcome} is indeed optimal among all that are consistent with $\phi^{-i}$.
\end{proof}

%\newpage

\section{Discussion}\label{sec:discussion}

Fully divisible investment and closed-loop strategies enable highly dynamic interaction. Theorems \ref{thm:static} and \ref{thm:dynamic} give differentiating answers to the question whether this implies that the strategic value of irreversibility eliminates the option value of waiting.

In (symmetric) open-loop equilibrium, the firms are still able to maximize their option values by investing at a constant price threshold, which is the same type of strategy that is optimal in monopoly.\footnote{The monopoly version of the application considered here was analyzed by \cite{Bertola98}. Concerning open-loop equilibrium, see \cite{Grenadier02} and \cite{BackPaulsen09}.}
Competition decreases the investment threshold and the option values, but only because the Cournot effect implies lower marginal revenue. The precommitment to open-loop strategies implies that preemption is no concern.

Theorem~\ref{thm:static} shows that closed-loop strategies, which indeed enable the firms to preempt each other's investments and expropriate the growth options, imply that it is no longer possible to support positive option values with such simple threshold rules; then the net present value of any investment is driven down to zero in equilibrium, like under perfect competition. 

However, the main driver for this result is that a constant price threshold makes the strategic preemption problem essentially static, since it does not affect the value of future growth options. Thus, the option to preempt any profitable investment is a rational one, but it is not impelling.

Indeed, there is a way out of the mutual preemption trap, by using more reactive strategies. Theorem~\ref{thm:dynamic} proves that it is enough to consider strategies that depend on the individual current capital stocks. It is not necessary to introduce history-dependent punishment regimes as in other dynamic games. The theorem shows that it is a credible threat to react to excessive investments by using a lower price threshold in the future. Hence, investing more than planned induces also the other firm to invest more---preemption will not be successful.

The parameter $c$ controls monotonically which option values are sustained in equilibrium. Thus, the equilibria can be Pareto-ranked by $c$. There is an upper limit that depends on the firms' initial capital stocks $(q_0^{1},q_0^{2})$ before the game starts. Low values of $c$ mean that the firms invest more aggressively, and in the limit the option values vanish again. 

In the equilibria with positive $c$, the strategic tradeoff between preemption and the option value of waiting shows much more distinctly. Investments have positive returns, deferred investments will be preempted, but further preemption is not worthwhile. In contrast, in the equilibrium with perfectly competitive outcome, the expected return of every invested unit is zero. Investment happens purely out of indifference and not to steal any profit. 

The latter behavior seems much less reasonable. It is similar to playing the preemptive equilibrium in a real option game with a single investment opportunity for each firm when there exists also an equilibrium of delayed simultaneous investment, as it was first noted by \cite{FudenbergTirole85} in a deterministic setting. However, the alternatives are quite different, because in the present setting the firms cannot simply coordinate on not preempting each other. Instead, the level to which preemption is moderated is determined by dynamic interaction.

%\newpage

\appendix
\renewcommand{\theequation}{\thesection\arabic{equation}}
%\section{Appendix}

\section{Proof of Theorem~\ref{thm:ver}}\label{app:ver}
\setcounter{equation}{0}

Recall the two following equivalences for any (finite) reflection boundary, which will be used repeatedly throughout the proof:
\begin{equation*}\tag{\ref{q<phi}}
q^{i}<\phi^{i}(x,q^{-i}) \iff x>\bar X^{i}(q^{i},q^{-i})
\end{equation*}
and
\begin{equation*}\tag{\ref{q=phi}}
q^{i}=\phi^{i}(x,q^{-i}) \iff x=\bar X^{i}(q^{i},q^{-i}).
\end{equation*}
%It follows that also the graph of the reflection boundary has the two equivalent representations
%\begin{equation}\label{boundary_graph}
%\{x=\bar X^{i}(q^{i},q^{-i})\} = \{q^{i}=\phi^{i}(x,q^{-i})\}.
%\end{equation}
Their first application is in proving the following key properties of outcomes that are consistent with reflection strategies.

\begin{lemma}\label{lem:reflection_outcome}
For any outcome that is consistent with $\phi^{i}$, i.e., which satisfies
\begin{equation*}
Q_t^{i}=q^{i}\vee\sup_{0\leq s\leq t}\phi^{i}(X_s,Q_s^{-i}),
\end{equation*}
the following hold true for all $t$:
\begin{enumerate}
\item\label{reflection_state} $X_t\leq\bar X^{i}(Q_t^{i},Q_t^{-i})$.
\item\label{reflection_cont} $Q_t^{i}$ is continuous in all $t>0$ and right-continuous in $t=0$.
\item\label{reflection_incr} $dQ_t^{i}>0 \implies X_t=\bar X^{i}(Q_t^{i},Q_t^{-i})$.
\end{enumerate}
\end{lemma}

\noindent{\itshape Remark.}\; As the proof of the lemma does not use right-continuity of $Q^{i}$ (or $Q^{-i}$), item \ref{reflection_cont} shows that it is in fact necessary.

\begin{proof}%[Proof of Lemma~\ref{lem:reflection_outcome}]
All items are clearly true if $\bar X^{i}=\infty$, because then $Q^{i}$ is constant. So suppose $\bar X^{i}$ is finite. 

\ref{reflection_state} This holds by $Q_t^{i}\geq\phi^{i}(X_t,Q_t^{-i})$ and equivalence \eqref{q<phi}. 

\ref{reflection_cont} First consider $t>0$. Then, for any $s<t$, $Q_{t-}^{i}\geq\phi^{i}(X_s,Q_s^{-i})$, so $X_s\leq\bar X^{i}(Q_{t-}^{i},Q_s^{-i})$ again by \eqref{q<phi}. This implies $X_t\leq\bar X^{i}(Q_{t-}^{i},Q_t^{-i})$ by continuity of $X$ in $t$ and the fact that $\bar X^{i}$ is nondecreasing in $q^{-i}$. Using the two latter facts once more and also the strict monotonicity of $\bar X^{i}$ in $q^{i}$, it follows that for every $\varepsilon>0$ there exists some $\delta>0$ such that $X_s\leq\bar X^{i}(Q_{t-}^{i}+\varepsilon,Q_s^{-i})$ for all $s\in[t,t+\delta)$. By \eqref{q<phi} it follows that $Q_{t-}^{i}+\varepsilon\geq Q_{t+}^{i}$. Letting $\varepsilon$ vanish yields continuity in $t$, because $Q^{i}$ is nondecreasing. For $t=0$, just skip the reasoning for $s<t$ and use $Q_t^{i}$ in place of $Q_{t-}^{i}$ to obtain right-continuity. 

\ref{reflection_incr} To prove the contrapositive, it is by \ref{reflection_state} enough to consider $X_t<\bar X^{i}(Q_t^{i},Q_t^{-i})$. Then continuity of $X_t$, the fact that $\bar X^{i}$ is nondecreasing in $q^{-i}$, and equivalences \eqref{q<phi} and \eqref{q=phi} imply that $Q_t^{i}>\phi^{i}(X_s,Q_s^{-i})$ for all $s$ in some nonempty interval $[t,t+\varepsilon)$, so that $Q^{i}$ must remain constant there.
\end{proof}

The equivalences \eqref{q<phi} and \eqref{q=phi} are also useful for proving the following implications concerning the partial derivatives of $V$, which will be needed to prove the theorem.

\begin{lemma}\label{lem:V_q}
The hypothesis that $V_{q^{i}}=1$ on $\{x\geq\bar X^{i}(q^{i},q^{-i})\}$ implies that in this region
\begin{equation*}
V_{q^{-i}}(x,q^{i},q^{-i})=V_{q^{-i}}(x,\phi^{i}(x,q^{-i}),q^{-i}).
\end{equation*}
The hypothesis that $V_{q^{-i}}=0$ on $\{\bar X^{i}(q^{i},q^{-i})>\bar X^{-i}(q^{-i},q^{i})\}\cap\{\bar X^{i}(q^{i},q^{-i})\geq x\geq\bar X^{-i}(q^{-i},q^{i})\}$ implies that, whenever $\bar X^{i}(q^{i},q^{-i})>x\geq\bar X^{-i}(q^{-i},q^{i})$,
\begin{equation*}
V_{q^{i}}(x,q^{i},q^{-i})=V_{q^{i}}(x,q^{i},\phi^{-i}(x,q^{i})).
\end{equation*}
\end{lemma}

\begin{proof}%[Proof of Lemma~\ref{lem:V_q}]
To prove the first claim, consider any state such that $x\geq\bar X^{i}(q^{i},q^{-i})$, which implies that $\bar X^{i}$ is finite. If $x=\bar X^{i}(q^{i},q^{-i})$, then $q^{i}=\phi^{i}(x,q^{-i})$ by equivalence \eqref{q=phi}, so the claim is true. Thus, assume $x>\bar X^{i}(q^{i},q^{-i})$. Then $q^{i}<\phi^{i}(x,q^{-i})$ by \eqref{q<phi}, and still $x\geq\bar X^{i}(q,q^{-i})$ for all $q\leq\phi^{i}(x,q^{-i})$ by \eqref{q<phi} and \eqref{q=phi}. Thus, $V_{q^{i}}(x,q,q^{-i})=1$ for all such $q$ by hypothesis, which implies that
\begin{equation*}
V(x,q^{i},q^{-i})=V(x,\phi^{i}(x,q^{-i}),q^{-i})-\phi^{i}(x,q^{-i})+q^{i}.
\end{equation*}
By continuity of $\bar X^{i}$, this equation holds in a neighborhood of the present state. Thus, using that $V_{q^{i}}(x,q,q^{-i})=1$ for $q=\phi^{i}(x,q^{-i})$ by hypothesis, it follows that
\begin{equation*}
V_{q^{-i}}(x,q^{i},q^{-i})=\partial_{q^{-i}}\left(V(x,\phi^{i}(x,q^{-i}),q^{-i})-\phi^{i}(x,q^{-i})+q^{i}\right)=V_{q^{-i}}(x,\phi^{i}(x,q^{-i}),q^{-i}),
\end{equation*}
noting that $\phi^{i}$ is differentiable in a neighborhood of $(x,q^{-i})$ by the implicit function theorem, because it presently takes a value $q>q^{i}\geq 0$, $\bar X^{i}(q,q^{-i})$ is continuously differentiable whenever $q>0$, and, by equivalence \eqref{q=phi}, $x=\bar X^{i}(q,q^{-i})$ if and only if $q=\phi^{i}(x,q^{-i})$.

To prove the second claim, consider any state such that $\bar X^{i}(q^{i},q^{-i})>x\geq\bar X^{-i}(q^{-i},q^{i})$, which now implies that $\bar X^{-i}$ is finite. As before, if $x=\bar X^{-i}(q^{-i},q^{i})$, then $q^{-i}=\phi^{-i}(x,q^{i})$, so the claim is true. Thus, assume $\bar X^{i}(q^{i},q^{-i})>x>\bar X^{-i}(q^{-i},q^{i})$. Then $q^{-i}<\phi^{-i}(x,q^{i})$ by \eqref{q<phi}, and still $\bar X^{i}(q^{i},q)>x\geq\bar X^{-i}(q^{-i},q)$ for all $q\in[q^{-i},\phi^{-i}(x,q^{i})]$ by \eqref{q<phi}, \eqref{q=phi}, and the fact that $\bar X^{i}$ is nondecreasing in $q^{-i}$. Thus, $V_{q^{-i}}(x,q^{i},q)=0$ for all such $q$ by hypothesis, which implies that
\begin{equation*}
V(x,q^{i},q^{-i})=V(x,q^{i},\phi^{-i}(x,q^{i})).
\end{equation*}
By continuity of $\bar X^{i}$ and $\bar X^{-i}$, this equation holds in a neighborhood of the present state. Thus, using that $V_{q^{-i}}(x,q^{i},q)=0$ for $q=\phi^{-i}(x,q^{i})$ by hypothesis, it follows that
\begin{equation*}
V_{q^{i}}(x,q^{i},q^{-i})=\partial_{q^{i}}V(x,q^{i},\phi^{-i}(x,q^{i}))=V_{q^{i}}(x,q^{i},\phi^{-i}(x,q^{i})),
\end{equation*}
noting that $\phi^{-i}$ is differentiable in a neighborhood of $(x,q^{i})$ for the same reason as $\phi^{i}$ in the proof of the first claim.
\end{proof}

Now, in order to prove the first claim of the theorem, fix an arbitrary state $(x,q^{1},q^{2})$, let $X_0=x$, and suppose $(Q^{1},Q^{2})$ is an outcome from $(\phi^{1},\phi^{2})$ that satisfies \eqref{V_q-i_0}--\eqref{limV=0} for some firm $i$. The effect of the initial investments on $V$ is
\begin{align*}
&V(X_0,Q_0^{i},Q_0^{-i})-V(X_0,q^{i},q^{-i}) \\
={}&V(X_0,Q_0^{i},Q_0^{-i})-V(X_0,Q_0^{i},q^{-i})+V(X_0,Q_0^{i},q^{-i})-V(X_0,q^{i},q^{-i}) \\
={}&\int_{q^{-i}}^{Q_0^{-i}}V_{q^{-i}}(X_0,Q_0^{i},q)\,dq+\int_{q^{i}}^{Q_0^{i}}V_{q^{i}}(X_0,q,q^{-i})\,dq.
\end{align*}
In the second integral, since the outcome is consistent with $\phi^{i}$, $Q_0^{i}>q^{i}$ only if $Q_0^{i}=\phi^{i}(X_0,Q_0^{-i})>q^{i}$, which implies that $\bar X^{i}$ is finite and, thus, $X_0=\bar X^{i}(Q_0^{i},Q_0^{-i})$ by \eqref{q=phi}. Then $X_0>\bar X^{i}(q,q^{-i})$ for all $q\in(q^{i},Q_0^{i})$, since $\bar X^{i}$ is strictly increasing in $q^{i}$ and nondecreasing in $q^{-i}$. Hence, $V_{q^{i}}=1$ inside the integral by condition~\ref{V_qi=1}, so its value is in any case $Q_0^{i}-q^{i}$.

In the first integral, analogously $X_0>\bar X^{-i}(q,Q_0^{i})$ for all $q\in(q^{-i},Q_0^{-i})$. Let $\tilde q=\sup\{q\geq 0\mid\bar X^{i}(Q_0^{i},q)<X_0\}\vee q^{-i}$. Then, since $\bar X^{i}(Q_0^{i},q)\geq X_0$ for all $q>\tilde q$, $V_{q^{-i}}=0$ inside the integral for all $q\in(\tilde q,Q_0^{-i})$ by condition~\ref{V_q-i=0}. For any $q<\tilde q$, $\bar X^{i}(Q_0^{i},q)<X_0$. Then the integrand is equal to $V_{q^{-i}}(X_0,\phi^{i}(X_0,q),q)$ by Lemma~\ref{lem:V_q}, where $\phi^{i}(X_0,q)>Q_0^{i}$ by \eqref{q<phi}. By \eqref{V_q-i_0} it follows that $V_{q^{-i}}=0$ also for all $q\in(q^{-i},\tilde q\vee Q_0^{-i})$, so the first integral vanishes entirely. Thus, in summary
\begin{equation}\label{DeltaV_0}
V(X_0,Q_0^{i},Q_0^{-i})-V(X_0,q^{i},q^{-i})=Q_0^{i}-q^{i}=\Delta Q_0^{i}.
\end{equation}

To evaluate the effect of the further evolution of the state on $V$, note that $X_t\leq\bar X^{i}(Q_t^{i},Q_t^{-i})\wedge\bar X^{-i}(Q_t^{-i},Q_t^{i})$ for all $t$ by Lemma~\ref{lem:reflection_outcome}, i.e., when starting from $Q_0^{i}$ and $Q_0^{-i}$, the state stays in a region in which $V$ is by condition~\ref{V_PDE} sufficiently differentiable to apply It\^{o}'s lemma. Since $Q^{i}$ and $Q^{-i}$ are continuous by Lemma~\ref{lem:reflection_outcome} when starting from $Q_0^{i}$ and $Q_0^{-i}$, this yields, for any $T\in\R_+$ and any stopping time $\tau$,
\begin{align*}
&e^{-r(\tau\wedge T)}V(X_{\tau\wedge T},Q_{\tau\wedge T}^{i},Q_{\tau\wedge T}^{-i})-V(X_0,Q_0^{i},Q_0^{-i}) \\
={}&\int_0^{\tau\wedge T}e^{-rt}\left(-rV(X_t,Q_t^{i},Q_t^{-i})+\mu(X_t)V_x(X_t,Q_t^{i},Q_t^{-i})+\frac12\sigma(X_t)^2V_{xx}(X_t,Q_t^{i},Q_t^{-i})\right)\,dt \\
+{}&\int_0^{\tau\wedge T}e^{-rt}\sigma(X_t)V_{x}(X_t,Q_t^{i},Q_t^{-i})\,dB_t \\
+{}&\int_{(0,\tau\wedge T]}e^{-rt}V_{q^{i}}(X_t,Q_t^{i},Q_t^{-i})\,dQ_t^{i} \\
+{}&\int_{(0,\tau\wedge T]}e^{-rt}V_{q^{-i}}(X_t,Q_t^{i},Q_t^{-i})\,dQ_t^{-i}.
\end{align*}
By condition~\ref{V_PDE}, the integral with respect to $dt$ equals $-\int_0^{\tau\wedge T}e^{-rt}\pi(X_t,Q_t^{i},Q_t^{-i})\,dt$. By Lemma~\ref{lem:reflection_outcome}, $dQ_t^{i}>0$ only when $X_t=\bar X^{i}(Q_t^{i},Q_t^{-i})$. Then $V_{q^{i}}=1$ by condition~\ref{V_qi=1}, so the value of the integral with respect to $dQ^{i}$ is $\int_{(0,\tau\wedge T]}e^{-rt}\,dQ_t^{i}$. Lemma~\ref{lem:reflection_outcome} further implies that $dQ_t^{-i}>0$ only when $X_t=\bar X^{-i}(Q_t^{-i},Q_t^{i})\leq\bar X^{i}(Q_t^{i},Q_t^{-i})$. When the latter inequality is strict, $V_{q^{-i}}=0$ by condition~\ref{V_q-i=0}, and when it binds, $V_{q^{-i}}=0$ by \eqref{V_q-i_t}. Thus, the integral with respect to $dQ^{-i}$ is zero.

The integral with respect to $dB$ is a local martingale. Consider a corresponding localizing sequence of stopping times $\tau_n\to\infty$, so that the integral is zero in expectation for every $\tau=\tau_n$. Thus,
\begin{align*}
V(X_0,Q_0^{i},Q_0^{-i})={}&E\left[\int_0^{\tau_n\wedge T}e^{-rt}\pi(X_t,Q_t^{i},Q_t^{-i})\,dt\right]-E\left[\int_{(0,\tau_n\wedge T]}e^{-rt}\,dQ_t^{i}\right] \\
+{}&E\left[e^{-r(\tau_n\wedge T)}V(X_{\tau_n\wedge T},Q_{\tau_n\wedge T}^{i},Q_{\tau_n\wedge T}^{-i})\right].
\end{align*}
By the assumption that $\pi$ is bounded from below, $dQ^{i}\geq 0$, and \eqref{Vintegrable}, it is possible to pass to the limit as $n\to\infty$, so that
\begin{align*}
V(X_0,Q_0^{i},Q_0^{-i})={}&E\left[\int_0^{T}e^{-rt}\pi(X_t,Q_t^{i},Q_t^{-i})\,dt\right]-E\left[\int_{(0,T]}e^{-rt}\,dQ_t^{i}\right] \\
+{}&E\left[e^{-rT}V(X_T,Q_T^{i},Q_T^{-i})\right].
\end{align*}
For the same reason as before, it is possible to pass to the limit in the two integrals as $T\to\infty$. By \eqref{limV=0} and \eqref{DeltaV_0}, it then follows that
\begin{equation*}
V(x,q^{i},q^{-i})=V(X_0,q^{i},q^{-i})=\Pi(Q^{i},Q^{-i}).
\end{equation*}

Next, to prove the second claim of the theorem, assume $V$ has also the additional properties.

\begin{lemma}\label{lem:V_q_leq}
The additional property that $V_{q^{i}}\leq 1$ on $\{x\leq\bar X^{1}(q^{1},q^{2})\wedge\bar X^{2}(q^{2},q^{1})\}$ implies that $V_{q^{i}}\leq 1$ for all states. The additional property that $V_{q^{-i}}\leq 0$ on $\{x=\bar X^{i}(q^{i},q^{-i})\leq\bar X^{-i}(q^{-i},q^{i})\}$ implies that $V_{q^{-i}}\leq 0$ on $\{x\geq\bar X^{-i}(q^{-i},q^{i})\}$.
\end{lemma}

\begin{proof}%[Proof of Lemma~\ref{lem:V_q_leq}]
Since $V_{q^{i}}=1$ on $\{x\geq\bar X^{i}(q^{i},q^{-i})\}$ by condition~\ref{V_qi=1}, consider the region $\{\bar X^{i}(q^{i},q^{-i})>x>\bar X^{-i}(q^{-i},q^{i})\}$ to prove the first claim. There, by Lemma~\ref{lem:V_q}, $V_{q^{i}}=V_{q^{i}}(x,q^{i},\phi^{-i}(x,q^{i}))$. To see that the latter is at most equal to one by the additional property, note that $x\leq\bar X^{-i}(\phi^{-i}(x,q^{i}),q^{i})$ by \eqref{q<phi}, and that also $x\leq\bar X^{i}(q^{i},\phi^{-i}(x,q^{i}))$ in the considered region, since $\phi^{-i}(x,q^{i})>q^{-i}$ by \eqref{q<phi} and $\bar X^{i}$ is nondecreasing in $q^{-i}$.

As to the second claim, recall that $V_{q^{-i}}=0$ on $\{\bar X^{i}(q^{i},q^{-i})>\bar X^{-i}(q^{-i},q^{i})\}\cap\{\bar X^{i}(q^{i},q^{-i})\geq x\geq\bar X^{-i}(q^{-i},q^{i})\}$ by condition~\ref{V_q-i=0}. This implies the claim if $\bar X^{i}=\infty$, so assume also $\bar X^{i}$ is finite. By condition~\ref{V_q-i=0} and the additional property, it is enough to show that $V_{q^{-i}}\leq 0$ on $\{x>\bar X^{i}(q^{i},q^{-i})\}$. There, $V_{q^{-i}}=V_{q^{-i}}(x,\phi^{i}(x,q^{-i}),q^{-i})$ by Lemma~\ref{lem:V_q}. But $\{q^{i}=\phi^{i}(x,q^{-i})\}=\{x=\bar X^{i}(q^{i},q^{-i})\}$ by \eqref{q=phi}, and indeed $V_{q^{-i}}\leq 0$ on $\{x=\bar X^{i}(q^{i},q^{-i})\}$ by condition~\ref{V_q-i=0} and the additional property.
\end{proof}

Now fix again an arbitrary state $(x,q^{1},q^{2})$ and let $X_0=x$, but consider any outcome that is consistent with $\phi^{-i}$ for some firm $i$ and satisfies \eqref{Vintegrable} and \eqref{limV=0}. The effect of the initial investments on $V$ is, as before,
\begin{align*}
&V(X_0,Q_0^{i},Q_0^{-i})-V(X_0,q^{i},q^{-i}) \\
={}&\int_{q^{-i}}^{Q_0^{-i}}V_{q^{-i}}(X_0,Q_0^{i},q)\,dq+\int_{q^{i}}^{Q_0^{i}}V_{q^{i}}(X_0,q,q^{-i})\,dq.
\end{align*}
In the first integral, since the outcome is consistent with $\phi^{-i}$, still $X_0>\bar X^{-i}(q,Q_0^{i})$ for all $q\in(q^{-i},Q_0^{-i})$, which implies that the integrand is at most equal to zero by Lemma~\ref{lem:V_q_leq}. Thus, because furthermore $V_{q^{i}}\leq 1$ for all states by Lemma~\ref{lem:V_q_leq}, now
\begin{equation}\label{DeltaV_0_leq}
V(X_0,Q_0^{i},Q_0^{-i})-V(X_0,q^{i},q^{-i})\leq Q_0^{i}-q^{i}=\Delta Q_0^{i}.
\end{equation}

Concerning the further evolution of the state, $X_t\leq\bar X^{-i}(Q_t^{-i},Q_t^{i})$ for all $t$ by Lemma~\ref{lem:reflection_outcome}. Thus, given the first additional property of $V$, the state still stays in a region in which $V$ is sufficiently differentiable to apply It\^{o}'s lemma when starting from $Q_0^{i}$ and $Q_0^{-i}$. Since $Q^{-i}$ is continuous by Lemma~\ref{lem:reflection_outcome} when starting from $Q_0^{-i}$, but $Q^{i}$ can now have jumps, this yields, for any $T\in\R_+$ and any stopping time $\tau$,
\begin{align*}
&e^{-r(\tau\wedge T)}V(X_{\tau\wedge T},Q_{\tau\wedge T}^{i},Q_{\tau\wedge T}^{-i})-V(X_0,Q_0^{i},Q_0^{-i}) \\
={}&\int_0^{\tau\wedge T}e^{-rt}\left(-rV(X_t,Q_{t-}^{i},Q_t^{-i})+\mu(X_t)V_x(X_t,Q_{t-}^{i},Q_t^{-i})+\frac12\sigma(X_t)^2V_{xx}(X_t,Q_{t-}^{i},Q_t^{-i})\right)\,dt \\
+{}&\int_0^{\tau\wedge T}e^{-rt}\sigma(X_t)V_{x}(X_t,Q_{t-}^{i},Q_t^{-i})\,dB_t \\
+{}&\int_{(0,\tau\wedge T]}e^{-rt}V_{q^{i}}(X_t,Q_{t-}^{i},Q_t^{-i})\,dQ_t^{i}-\sum_{0<t\leq\tau\wedge T}e^{-rt}V_{q^{i}}(X_t,Q_{t-}^{i},Q_t^{-i})\,\Delta Q_t^{i} \\
+{}&\int_{(0,\tau\wedge T]}e^{-rt}V_{q^{-i}}(X_t,Q_{t-}^{i},Q_t^{-i})\,dQ_t^{-i} \\
+{}&\sum_{0<t\leq\tau\wedge T}e^{-rt}\Delta V(X_t,Q_t^{i},Q_t^{-i}).
\end{align*}
By condition~\ref{V_PDE_leq}, the integral with respect to $dt$ is at most equal to $-\int_0^{\tau\wedge T}e^{-rt}\pi(X_t,Q_{t-}^{i},Q_t^{-i})\,dt$, which is equal to $-\int_0^{\tau\wedge T}e^{-rt}\pi(X_t,Q_t^{i},Q_t^{-i})\,dt$ since $Q^{i}$ has at most countably many jumps. Since $dQ^{i}_t\geq\Delta Q_t^{i}$, and $V_{q^{i}}\leq 1$ by Lemma~\ref{lem:V_q_leq}, the value of the difference between the integral with respect to $dQ^{i}$ and the adjacent sum is at most equal to $\int_{(0,\tau\wedge T]}e^{-rt}\,dQ_t^{i}-\sum_{0<t\leq\tau\wedge T}e^{-rt}\,\Delta Q_t^{i}$. By Lemma~\ref{lem:reflection_outcome}, $dQ_t^{-i}>0$ only when $X_t=\bar X^{-i}(Q_t^{-i},Q_t^{i})$. Then $V_{q^{-i}}\leq 0$ by Lemma~\ref{lem:V_q_leq}, so the value of the integral with respect to $dQ^{-i}$ is nonpositive. Since $X$ and $Q^{-i}$ are continuous in any $t>0$, $\Delta V(X_t,Q_t^{i},Q_t^{-i})=\int_{Q_{t-}^{i}}^{Q_t^{i}}V_{q^{i}}(X_t,q,Q_t^{-i})\,dq$, which by Lemma~\ref{lem:V_q_leq} is at most equal to $Q_t^{i}-Q_{t-}^{i}=\Delta Q_t^{i}$.

Imposing all these estimates, the sums cancel, and the remaining terms are the same that appeared in the proof of the first claim of the theorem. Therefore, it is possible to follow the same steps from this point on, just with an inequality instead of the equality and using \eqref{DeltaV_0_leq} instead of \eqref{DeltaV_0}. This leads to the conclusion that
\begin{flalign*}
&& &V(x,q^{i},q^{-i})=V(X_0,q^{i},q^{-i})\geq\Pi(Q^{i},Q^{-i}). &&\mathllap{\qed}
\end{flalign*}

\section{Lemmas for the proof of Theorem~\ref{thm:static}}\label{app:static}
\setcounter{equation}{0}

\begin{lemma}\label{lem:Qintegrable}
For any right-continuous and nondecreasing process $Q$ with $Q_0\geq 0$, and any $r>0$, it holds true that
\begin{equation*}
E\left[Q_0\right]+E\left[\int_{(0,\infty)}e^{-rt}\,dQ_t\right]=E\left[r\int_0^{\infty} e^{-rt}Q_t\,dt\right]
\end{equation*}
and, if either side is finite, then
\begin{equation}\label{Qintegrable}
\forall T>0:\quad E\left[Q_T\right]<\infty
%x_0>\bar X^{i}(Q^{i},q^{-i})
\end{equation}
and
\begin{equation}\label{limQ=0}
\lim_{T\to\infty}E\left[e^{-rT}Q_T\right]=0.
\end{equation}
\end{lemma}

\begin{proof}%[Proof of Lemma~\ref{lem:Qintegrable}]
Integrating $e^{-rt}Q_t$ by parts yields, for any finite $T$,
\begin{equation*}
Q_0+\int_{(0,T]}e^{-rt}\,dQ_t=r\int_0^T e^{-rt}Q_t\,dt+e^{-rT}Q_T,
\end{equation*}
where $0\leq e^{-rT}Q_T\leq r\int_T^{\infty} e^{-rt}Q_t\,dt$ since $Q_0\geq 0$, $Q$ is nondecreasing, and $r>0$. From this, the first claimed equation and, if either side is finite, also \eqref{limQ=0} follow by monotone convergence. Moreover, \eqref{limQ=0} and the fact that $Q$ is nondecreasing imply \eqref{Qintegrable}.
\end{proof}

\begin{lemma}\label{lem:p_admissible}
Let $\phi^{-i}$ be a reflection strategy that corresponds to a constant price threshold $p>0$. Then there is a unique outcome such that firm $i$ does not invest at all and that is consistent with $\phi^{-i}$. Further, this outcome is consistent with any reflection strategy $\phi^{i}$ that uses a boundary $\bar X^{i}\geq\bar X^{-i}$.
\end{lemma}

%\footnote{The condition $\gamma<\beta$ is also necessary for $Q^{i}$ to have finite investment cost, as can be seen from the proof.}

\begin{proof}%[Proof of Lemma~\ref{lem:p_admissible}]
There is only one possible outcome such that firm $i$ does not invest at all and that is consistent with $\phi^{-i}$, because these requirements fully determine the only candidate by $Q_t^{i}=q^{i}$ and $Q_t^{-i}=q^{-i}\vee\sup_{0\leq s\leq t}\phi^{-i}(X_s,q^{i})$. This candidate clearly satisfies all conditions for an admissible outcome except for finiteness of the investment cost for $Q^{-i}$. Putting the latter aside, the candidate is consistent with any reflection strategy $\phi^{i}$ using a boundary $\bar X^{i}\geq\bar X^{-i}$ by the following chain of implications. Consistency with $\phi^{-i}$ implies $Q_t^{-i}\geq\phi^{-i}(X_t,q^{i})$, which is equivalent to $X_t\leq\bar X^{-i}(Q_t^{-i},q^{i})$ by \eqref{q<phi}, which by $\bar X^{i}\geq\bar X^{-i}$ implies $X_t\leq\bar X^{i}(q^{i},Q_t^{-i})$, which is again equivalent to $q^{i}\geq\phi^{i}(X_t,Q_t^{-i})$ by \eqref{q<phi}; and since this holds for all $t$, it follows that $q^{i}\vee\sup_{0\leq s\leq t}\phi^{i}(X_s,Q_s^{-i})=q^{i}=Q_t^{i}$. Therefore, it only remains to show that the investment cost for $Q^{-i}$ is finite. 

By Lemma~\ref{lem:Qintegrable}, the investment cost is equal to $rE[\int_0^{\infty}e^{-rt}Q_t^{-i}\,dt]-q^{-i}$, where presently $Q_t^{-i}=q^{-i}\vee\sup_{0\leq s\leq t}\phi^{-i}(X_s,q^{i})$ and $\phi^{-i}(x,q^{i})=(x/p)^{\gamma}-q^{i}$. Thus, the investment cost is finite if and only if
\begin{equation}\label{integrable_Z}
E\left[\int_0^{\infty}e^{-rt}\sup_{0\leq s\leq t}X_s^{\gamma}\,dt\right]<\infty.
\end{equation}
Since $X_t^{\gamma}=X_0^{\gamma}\exp((\mu_{\gamma}-\frac12\sigma_{\gamma}^2)t+\sigma_{\gamma}B_t)$ for $\mu_{\gamma}=\gamma\mu+\gamma(\gamma-1)\frac12\sigma^2$ and $\sigma_{\gamma}=\gamma\sigma$, \eqref{integrable_Z} holds if and only if $r>\mu_{\gamma}$ \citep[see, e.g.,][Chapter VII]{Bertoin96}, which is the assumed inequality \eqref{r>mu_gamma}.
\end{proof}

\begin{lemma}\label{lem:Vinfty_qi}
Fix any $p>0$ and let $\bar X^{-i}=p/P$. Then the function $V^{\text{abs}}$ defined in \eqref{Vinfty} has the properties that
\begin{enumerate}
\item $V^{\text{abs}}<\frac{p}{p^{*}}q^{i}$ for all $x<\bar X^{-i}$ and $q^{i}>0$,
\item $V^{\text{abs}}_{q^{i}}$ is strictly increasing in $x\leq\bar X^{-i}$, and
\item $V^{\text{abs}}_{q^{i}}=\frac{p}{p^{*}}$ for all $x\geq\bar X^{-i}$.
\end{enumerate}
\end{lemma}

\begin{proof}%[Proof of Lemma~\ref{lem:Vinfty_qi}]
The first property is a consequence of the second, but it also follows directly from \eqref{Vinfty} and the facts that $xP(q^{i}+q^{-i})/p<1$ for $x<\bar X^{-i}$ and that the function $y-\beta^{-1}y^{\beta}$ is strictly increasing on $\{0\leq y\leq 1\}$ by $\beta>1$. To show the other two properties, differentiate the term in the definition of $V^{\text{abs}}$ for $x\leq\bar X^{-i}$ with respect to $q^{i}$, which yields
\begin{equation}\label{Vinfty_qi}
\frac{p}{r-\mu}\left(\frac{x(P+q^{i}P')}{p}-\frac{1}{\beta}\left(\frac{xP}{p}\right)^{\beta}-\left(\frac{xP}{p}\right)^{\beta-1}\frac{xq^{i}P'}{p}\right),
\end{equation}
where $P'$ indeed exists since $0<x\leq\bar X^{-i}$ implies $q^{i}+q^{-i}>0$. A further differentiation with respect to $x$ yields
\begin{align*}
&\frac{p}{r-\mu}\left(\frac{P+q^{i}P'}{p}-\left(\frac{xP}{p}\right)^{\beta-1}\frac{P+\beta q^{i}P'}{p}\right) \\
\geq{}&\frac{p}{r-\mu}\left(\frac{P+q^{i}P'}{p}\right)\left(1-\left(\frac{xP}{p}\right)^{\beta-1}\right),
\end{align*}
where the inequality is due to $r>\mu$, $\beta>1$, and $q^{i}P'\leq 0$. The last obtained term is strictly positive for any $x<\bar X^{-i}$, however, because then $xP/p<1$ and
\begin{equation*}
P+q^{i}P'=P\left(1-\frac{q^{i}}{\gamma(q^{i}+q^{-i})}\right)>0
\end{equation*}
by $P(q)=q^{-\frac{1}{\gamma}}$ and $\gamma>1$. This implies the second property. The third property is immediate from \eqref{Vinfty} for any $x>\bar X^{-i}$, and for $x=\bar X^{-i}$ it follows from \eqref{Vinfty_qi}, where then $xP/p=1$.
\end{proof}

\begin{lemma}\label{lem:Vinfty}
Let $\phi^{-i}$ be a reflection strategy that corresponds to a constant price threshold $p>0$. Then firm $i$'s payoff from not investing at all is equal to the value of the function $V^{\text{abs}}$ defined in \eqref{Vinfty}, and this choice is optimal if $p\leq p^{*}$.
\end{lemma}

\begin{proof}%[Proof of Lemma~\ref{lem:Vinfty}]
By Lemma~\ref{lem:p_admissible}, there exists a unique outcome such that firm $i$ does not invest at all and that is consistent with $\phi^{-i}$, and this outcome is also consistent with the reflection strategy $\phi^{i}$ that uses the boundary $\bar X^{i}=\infty$. Hence, it is possible to apply Theorem~\ref{thm:ver} using this pair of strategies and the function $V=V^{\text{abs}}$.

Then conditions \ref{V_PDE} and \ref{V_q-i=0} hold by construction, and condition \ref{V_qi=1} is void by $\bar X^{i}=\infty$. Also the implications in \eqref{V_q-i_0} and \eqref{V_q-i_t} are void, since $\bar X^{i}=\infty$ implies that $\phi^{i}=0$ and $x<\bar X^{i}$ for all states.  

Already in view of the second claim, the integrability requirements \eqref{Vintegrable} and \eqref{limV=0} will now be shown to hold for any outcome that is consistent with $\phi^{-i}$, which in particular covers the one in which firm $i$ does not invest. To do so, it will be enough to show that $\abs{V^{\text{abs}}}$ is bounded by a linear function of $q^{i}$, because Lemma~\ref{lem:Qintegrable} implies that \eqref{Qintegrable} and \eqref{limQ=0} hold for $Q_T=Q_T^{i}$ by finiteness of the investment cost for $Q^{i}$. This will then yield \eqref{Vintegrable} (by the fact that $Q^{i}$ is nondecreasing) and \eqref{limV=0}. Since $X_t\leq\bar X^{-i}(Q_t^{-i},Q_t^{i})$ for all $t\geq 0$ by Lemma~\ref{lem:reflection_outcome}, it is enough to consider only $x\leq\bar X^{-i}$ to derive the bound on $\abs{V^{\text{abs}}}$. Then $0\leq xP/p\leq 1$, which implies that indeed
\begin{equation*}
\absd{V^{\text{abs}}}\leq\frac{p}{r-\mu}\left(1+\frac{1}{\beta}\right)q^{i}.
\end{equation*}
This completes the proof of the first claim, and all that remains to prove the second claim is to verify that $V=V^{\text{abs}}$ has the three additional properties in Theorem~\ref{thm:ver} if $p\leq p^{*}$. 

Given that $\bar X^{i}=\infty$, condition \ref{V_PDE_leq} is weaker than condition \ref{V_PDE}, and condition \ref{V_q-i<0} is void. Thus, the restriction $p\leq p^{*}$ is needed only for condition \ref{V_qi<1}, which then holds by Lemma~\ref{lem:Vinfty_qi}.
\end{proof}

\begin{lemma}\label{lem:Vp}
Let $\phi^{1}$ and $\phi^{2}$ be reflection strategies that correspond to a constant price threshold $p>0$. Then firm $i$'s payoff from the outcome such that the other firm does not invest at all is equal to the value of the function $V^{\text{inv}}$ defined in \eqref{Vp}.
\end{lemma}

\begin{proof}%[Proof of Lemma~\ref{lem:Vp}]
By switching roles in Lemma~\ref{lem:p_admissible}, there indeed exists a unique outcome from $(\phi^{1},\phi^{2})$ such that the other firm does not invest at all. Consider this outcome and the function $V=V^{\text{inv}}$ in Theorem~\ref{thm:ver}. Then conditions \ref{V_PDE} and \ref{V_qi=1} hold by construction, and condition \ref{V_q-i=0} is void by $\bar X^{i}=\bar X^{-i}$. Also the implications in \eqref{V_q-i_0} and \eqref{V_q-i_t} are void, since $Q_t^{-i}=q^{-i}$ for all $t$.  

To show that the integrability requirements \eqref{Vintegrable} and \eqref{limV=0} hold for the outcome with constant $Q^{-i}$, it will be enough to show that $\abs{V^{\text{inv}}}$ is bounded by an affine function of $q^{i}$ for fixed $q^{-i}$, because Lemma~\ref{lem:Qintegrable} implies that \eqref{Qintegrable} and \eqref{limQ=0} hold for $Q_T=Q_T^{i}$ by finiteness of the investment cost for $Q^{i}$. This will then yield \eqref{Vintegrable} (by the fact that $Q^{i}$ is nondecreasing) and \eqref{limV=0} (by $r>0$). Since $X_t\leq\bar X^{-i}(Q_t^{-i},Q_t^{i})$ for all $t\geq 0$ by Lemma~\ref{lem:reflection_outcome}, it is enough to consider only $x\leq\bar X^{-i}$ to derive the bound on $\abs{V^{\text{inv}}}$. Then $0\leq xP/p\leq 1$, which implies that
\begin{equation*}
\absd{V^{\text{inv}}}\leq\frac{p}{r-\mu}q^{i}+\frac{\gamma}{\beta-\gamma}\left(\absd{\frac{p(\gamma-1)}{(r-\mu)\gamma}-1}q^{i}+\absd{\frac{p(\beta-1)}{(r-\mu)\beta}-1}q^{-i}\right). \qedhere
\end{equation*}
\end{proof}

\begin{lemma}\label{lem:Vp>Vinfty}
For any $p>0$,% and let $\bar X^{-i}=p/P$. Then%, for any state with $x>0$,
\begin{equation*}
V^{\text{inv}}\lessgtr V^{\text{abs}} \iff p\lessgtr p^{*}.
\end{equation*}
\end{lemma}

\begin{proof}%[Proof of Lemma~\ref{lem:Vp>Vinfty}]
First consider any state such that $0<xP\leq p$. Then $q^{i}+q^{-i}>0$ and, by the definitions of $V^{\text{inv}}$ and $V^{\text{abs}}$ in \eqref{Vp} and \eqref{Vinfty},
\begin{flalign*}
\left(V^{\text{inv}}-V^{\text{abs}}\right)\left(\frac{xP}{p}\right)^{-\beta}=\frac{\gamma}{\beta-\gamma}\left(\left(\frac{p(\gamma-1)}{(r-\mu)\gamma}-1\right)q^{i}+\left(\frac{p(\beta-1)}{(r-\mu)\beta}-1\right)q^{-i}\right)+\frac{p}{(r-\mu)\beta}q^{i}.
\end{flalign*}
On the left-hand side, $xP/p>0$, whereas the right-hand side is strictly increasing in $p$ by $q^{i}+q^{-i}>0$, and it takes the value zero at $p=p^{*}$. This implies the claimed equivalence.

Now consider any state such that $xP>p$, and let $\phi^{i}$ be the reflection strategy such that $\bar X^{i}=p/P$. Then
\begin{align*}
V^{\text{inv}}(x,q^{i},q^{-i})-V^{\text{abs}}(x,q^{i},q^{-i})={}&V^{\text{inv}}(x,\phi^{i},q^{-i})-V^{\text{abs}}(x,\phi^{i},q^{-i}) \\
&+\int_{q^{i}}^{\phi^{i}}\left(V_{q^{i}}^{\text{abs}}(x,q,q^{-i})-V_{q^{i}}^{\text{inv}}(x,q,q^{-i})\right)\,dq,
\end{align*}
where $q^{i}<\phi^{i}$ and $x\leq\bar X^{i}(\phi^{i},q^{-i})=p/P(\phi^{i}+q^{-i})$ by equivalence \eqref{q<phi}. Thus, the state in the first difference on the right-hand side is one for which the claim has already been proved. Inside the integral, for any $q<\phi^{i}$, still $x>\bar X^{i}(q,q^{-i})=p/P(q+q^{-i})$ by equivalence \eqref{q<phi}, so the integrand is equal to $p/p^{*}-1$ by \eqref{Vp} and \eqref{Vinfty}. It follows that the value of the integral is greater (less) than zero if and only if $p$ is greater (less) than $p^{*}$, and hence the same holds for the whole equation.
\end{proof}

\begin{lemma}\label{lem:indifference}
Consider any set of outcomes for some given state that all have the same aggregate capital $Q^{1}+Q^{2}$. If the set contains an element that maximizes each firm's profit (over the set), then each firm's profit is constant on the set.
\end{lemma}

\begin{proof}%[Proof of Lemma~\ref{lem:indifference}]
Since $\pi(x,q^{i},q^{-i})+\pi(x,q^{-i},q^{i})$ is a function of $x$ and $q^{i}+q^{-i}$ in the present setting, all outcomes from the given set yield the same aggregate profit $\Pi(Q^{1},Q^{2})+\Pi(Q^{2},Q^{1})$. Now suppose $(\tilde Q^{1},\tilde Q^{2})$ maximizes each firm's profit over the set. That is, for any other element $(Q^{1},Q^{2})$,
\begin{equation*}
\Pi(\tilde Q^{1},\tilde Q^{2})\geq \Pi(Q^{1},Q^{2})
\end{equation*}
and
\begin{equation*}
\Pi(\tilde Q^{2},\tilde Q^{1})\geq \Pi(Q^{2},Q^{1}).
\end{equation*}
But since the sum of the left-hand sides is equal to the sum of the right-hand sides, both inequalities must hold with equality.
\end{proof}

\begin{lemma}\label{lem:p_eql_payoff}
Suppose $(\phi^{1},\phi^{2})$ is an equilibrium consisting of two reflection strategies that correspond to a constant price threshold $p>0$. Then the equilibrium payoffs are equal to the value of the function $V^{\text{abs}}$.
\end{lemma}

\begin{proof}%[Proof of Lemma~\ref{lem:p_eql_payoff}]
Let $(Q^{1},Q^{2})$ be an equilibrium outcome, and consider $\tilde Q^{1}=Q^{1}+Q^{2}-q^{2}$ and $\tilde Q^{2}=q^{2}$. Then also $(\tilde Q^{1},\tilde Q^{2})$ is an outcome, since all required properties are inherited from $(Q^{1},Q^{2})$, and by construction $\tilde Q^{1}+\tilde Q^{2}=Q^{1}+Q^{2}$. Further, $(\tilde Q^{1},\tilde Q^{2})$ inherits consistency with $\phi^{2}$, because $Q_t^{2}=q^{2}\vee\sup_{0\leq s\leq t}\phi^{2}(X_s,Q_s^{1})$ implies that $Q_t^{2}\geq\phi^{2}(X_t,Q_t^{1})$, which by \eqref{q<phi} implies $X_t\leq\bar X^{2}(Q_t^{2},Q_t^{1})=p/P(Q_t^{1}+Q_t^{2})=p/P(\tilde Q_t^{1}+q^{2})=\bar X^{2}(q^{2},\tilde Q_t^{1})$, so that, again by \eqref{q<phi}, $q^{2}\geq\phi^{2}(X_t,\tilde Q_t^{1})$. And since this holds for all $t$, indeed $\tilde Q_t^{2}=q^{2}\vee\sup_{0\leq s\leq t}\phi^{2}(X_s,\tilde Q_s^{1})$. Thus, firm $1$'s equilibrium payoff $\Pi(Q^{1},Q^{2})$ must be at least equal to $\Pi(\tilde Q^{1},\tilde Q^{2})$.

By Lemma~\ref{lem:p_admissible}, there exists a unique outcome such that firm $2$ does not invest and that is also consistent with $\phi^{1}$. Let $(\hat Q^{1},\hat Q^{2})$ be this outcome, so $\hat Q_t^{1}=q^{1}\vee\sup_{0\leq s\leq t}\phi^{1}(X_s,q^{2})$ and $\hat Q_t^{2}=q^{2}$ for all $t$. Hence, to show next that $\tilde Q^{1}\geq\hat Q^{1}$, it is enough to show that $\tilde Q_t^{1}\geq\phi^{1}(X_t,q^{2})$ for all $t$, because $\tilde Q^{1}$ is nondecreasing and $\tilde Q^{1}\geq q^{1}$. But it was already argued above that $X_t\leq\bar X^{2}(q^{2},\tilde Q_t^{1})$, which indeed implies $\tilde Q_t^{1}\geq\phi^{1}(X_t,q^{2})$ by $\bar X^{2}(q^{2},\tilde Q_t^{1})=\bar X^{1}(\tilde Q_t^{1},q^{2})$ and \eqref{q<phi}.

Thus, as $\tilde Q^{2}=\hat Q^{2}$ and $\pi$ is nonincreasing in $q^{-i}$, it follows that $\Pi(\hat Q^{2},\hat Q^{1})\geq\Pi(\tilde Q^{2},\tilde Q^{1})$. Moreover, firm $2$'s equilibrium payoff is at least equal to $\Pi(\hat Q^{2},\hat Q^{1})$ by consistency of $(\hat Q^{1},\hat Q^{2})$ with $\phi^{1}$, so that together also $\Pi(Q^{2},Q^{1})\geq\Pi(\tilde Q^{2},\tilde Q^{1})$.

However, it was already shown above that $\Pi(Q^{1},Q^{2})\geq\Pi(\tilde Q^{1},\tilde Q^{2})$, so that both inequalities must in fact hold with equality by Lemma~\ref{lem:indifference}. Thus, also the previous inequalities $\Pi(Q^{2},Q^{1})\geq\Pi(\hat Q^{2},\hat Q^{1})\geq\Pi(\tilde Q^{2},\tilde Q^{1})$ must hold with equality, where the value of $\Pi(\hat Q^{2},\hat Q^{1})$ is indeed given by the function $V^{\text{abs}}$ by Lemma~\ref{lem:Vinfty}. This completes the proof for firm $2$, and repeating it with switched roles proves the claim also for firm $1$.
\end{proof}

\begin{lemma}\label{lem:reflection_sum}
%Let $\phi^{-i}$ be a reflection strategy with a boundary $\bar X^{-i}=f(q^{i}+q^{-i})$ for some function $f$. Then $f$ has an inverse function $f^{-1}$,
%\begin{equation}\label{reflection_sum_phi}
%\phi^{-i}(x,q^{i})=f^{-1}(x)-q^{-i},
%\end{equation}
%and every outcome that satisfies
%\begin{equation}\label{reflection_sum}
%Q_t^{1}+Q_t^{2}=(q^{1}+q^{2})\vee\sup_{0\leq s\leq t}f^{-1}(X_s)
%\end{equation}
%is consistent with $\phi^{-i}$.
Let $\phi^{1}$ and $\phi^{2}$ be reflection strategies that correspond to a constant price threshold $p>0$. Then the two outcomes from $(\phi^{1},\phi^{2})$ such that one firm does not invest have the same aggregate capital, and every other outcome with the same sum is also an outcome from $(\phi^{1},\phi^{2})$.
\end{lemma}

\begin{proof}%[Proof of Lemma~\ref{lem:reflection_sum}]
%Since any reflection boundary $\bar X^{-i}$ must be continuous and strictly increasing in $q^{-i}$, and be unbounded as $q^{-i}\to\infty$, the function $f$ necessarily has the analogous properties and, thus, an inverse $f^{-1}$ such that
%\begin{equation*}
%q<f^{-1}(x)\iff f(q)<x.
%\end{equation*}
%This equivalence implies \eqref{reflection_sum_phi} by $\bar X^{-i}=f(q^{i}+q^{-i})$. Thus, for any outcome satisfying \eqref{reflection_sum}, $Q_t^{-i}\geq f^{-1}(X_t)-Q_t^{i}=\phi^{-i}(X_t,Q_t^{i})$ for all $t$. Since $Q^{-i}$ is nondecreasing by hypothesis, and $Q_t^{-i}\geq q^{-i}$, it follows that $Q_t^{-i}\geq q^{-i}\vee\sup_{0\leq s\leq t}\phi^{-i}(X_s,Q_s^{i})$. Now showing consistency with $\phi^{-i}$ means showing that the latter inequality cannot be strict, so suppose by way of contradiction it is. Then $Q_t^{-i}>\sup_{0\leq s\leq t}f^{-1}(X_s)-Q_s^{i}\geq\sup_{0\leq s\leq t} f^{-1}(X_s)-Q_t^{i}$ since $Q^{i}$ is nondecreasing by hypothesis. But this would contradict \eqref{reflection_sum}, because by hypothesis now also $Q_t^{-i}>q^{-i}$, and $Q_t^{i}\geq q^{i}$.
For each $i$, the constant price threshold $\bar X^{i}=p/P$ implies that $\phi^{i}=P^{-1}(p/x)-q^{-i}$. Thus, it follows from consistency with $\phi^{i}$ that the outcome in which $Q_t^{-i}=q^{-i}$ for all $t$ satisfies the equation
\begin{equation}\label{reflection_sum}
Q_t^{1}+Q_t^{2}=(q^{1}+q^{2})\vee\sup_{0\leq s\leq t}P^{-1}(p/X_s),
\end{equation}
which fully determines aggregate capital. Now consider any other outcome such that \eqref{reflection_sum} holds. Then $Q_t^{i}\geq P^{-1}(p/X_t)-Q_t^{-i}=\phi^{i}(X_t,Q_t^{-i})$ for all $t$. Since $Q^{i}$ is nondecreasing and $Q_t^{i}\geq q^{i}$ for any outcome, it follows that $Q_t^{i}\geq q^{i}\vee\sup_{0\leq s\leq t}\phi^{i}(X_s,Q_s^{-i})$. Consistency with $\phi^{i}$ means the latter inequality always holds with equality, so assume by way of contradiction it is strict for some $t$. Then $Q_t^{i}>\sup_{0\leq s\leq t}P^{-1}(p/X_s)-Q_s^{-i}\geq\sup_{0\leq s\leq t} P^{-1}(p/X_s)-Q_t^{-i}$, since $Q^{-i}$ is nondecreasing for any outcome. But this contradicts \eqref{reflection_sum}, because also $Q_t^{i}>q^{i}$ by assumption, and $Q_t^{-i}\geq q^{-i}$. Thus, the outcome is indeed consistent with each $\phi^{i}$.
\end{proof}

\section{Lemmas for the proof of Theorem~\ref{thm:dynamic}}\label{app:dynamic}
\setcounter{equation}{0}

\begin{lemma}\label{lem:barXc}
Fix any $c\geq 0$ and consider the function $\bar X^{i}$ given by \eqref{barXc}.
%\begin{equation*}%\label{barXc}
%\bar X^{i}(q^{i},q^{-i})=\left(p^{*}+\frac{c}{q^{i}\vee q^{-i}}\right)\left(P(q^{i}+q^{-i})\right)^{-1},\qquad q^{i},q^{-i}\geq c\frac{2\gamma-1}{p^{*}}.
%\end{equation*}
Then $\bar X^{i}$ is strictly increasing in both arguments, and $\lim_{q^{i}\to\infty}\bar X^{i}(q^{i},q^{-i})=\infty$.
\end{lemma}

\begin{proof}%[Proof of Lemma~\ref{lem:barXc}]
Since $p^{*}>0$ and $c\geq 0$, it is clear that $\bar X^{i}\to\infty$ as $q^{i}\to\infty$. To show that $\bar X^{i}$ is strictly increasing on the given domain, which is obviously true for $c=0$ by $P(q)=q^{-1/\gamma}$, assume $c>0$. By symmetry, it is enough to show that $\bar X^{i}$ is strictly increasing in $q^{i}$, and by continuity it is enough to show this for the two cases $q^{i}\geq q^{-i}$ and $q^{i}<q^{-i}$. Therefore, consider first $q^{i}\geq q^{-i}\geq c(2\gamma-1)/p^{*}$. Then
\begin{equation*}
\bar X^{i}_{q^{i}}=\frac{1}{\gamma (q^{i})^2}\left(p^{*}(q^{i})^2+(1-\gamma)cq^{i}-\gamma cq^{-i}\right)\left(q^{i}+q^{-i}\right)^{\frac{1}{\gamma}-1}.
\end{equation*}
Since $\gamma cq^{-i}>0$ in the quadratic polynomial, it follows that $\bar X^{i}_{q^{i}}$ is strictly positive for all $q^{i}>q^{-i}$ if it is nonnegative for $q^{i}=q^{-i}$, which is indeed true by $q^{-i}\geq c(2\gamma-1)/p^{*}$. Thus, $\bar X^{i}$ is strictly increasing in $q^{i}\geq q^{-i}$. Next, consider the case $q^{-i}>q^{i}\geq c(2\gamma-1)/p^{*}$. Then clearly
\begin{equation*}
\bar X^{i}_{q^{i}}=\frac{1}{\gamma}\left(p^{*}+\frac{c}{q^{-i}}\right)\left(q^{i}+q^{-i}\right)^{\frac{1}{\gamma}-1}>0. \qedhere
\end{equation*}
\end{proof}

\begin{lemma}\label{lem:c_eql_outcome}
Let $(\phi^{1},\phi^{2})$ be a pair of reflection strategies with boundaries given by \eqref{barXc} for some $c\geq 0$, and consider any state such that $q^{1},q^{2}\geq c(2\gamma-1)/p^{*}$. Then the pair of processes $(Q^{1},Q^{2})$ given by \eqref{c_eql_outcome} is an outcome from $(\phi^{1},\phi^{2})$.
\end{lemma}

\begin{proof}%[Proof of Lemma~\ref{lem:c_eql_outcome}]
Let $Q^{i}$ be given by \eqref{c_eql_outcome} for each firm $i$, which is clearly admissible if the investment cost is finite. To see that this is the case, note that $Q^{i}$ is dominated by the capital that firm $i$ would accumulate if the other firm did not invest at all, $q^{i}\vee\sup_{0\leq s\leq t}\phi^{i}(X_s,q^{-i})$. The latter, furthermore, by definition of any reflection strategy $\phi^{i}$, is itself dominated by firm $i$'s capital in the outcome that is consistent with firm $i$ using the reflection strategy with constant price threshold $p=p^{*}$ and such that the other firm does not invest, because the hypothesis that $c\geq 0$ implies that $\bar X^{i}(q,q^{-i})\geq p^{*}(q+q^{-i})^{\frac{1}{\gamma}}=p^{*}/P(q+q^{-i})$ for any $q\geq 0$. This last outcome, consistent with a constant price threshold, is admissible by Lemma~\ref{lem:p_admissible} and, thus, has finite investment cost. By Lemma~\ref{lem:Qintegrable}, it then follows that also the investment cost for the present $Q^{i}$ is finite.

For proving that the present outcome $(Q^{1},Q^{2})$ is consistent with both reflection strategies, note the following two equivalences analogous to \eqref{q<phi} and \eqref{q=phi}, which hold for any $q\geq c(2\gamma-1)/p^{*}$ since $\bar X^{i}(q,q)$ is strictly increasing and continuous on this domain:
\begin{equation}\label{q<psi}
q<\psi^{c}(x)\iff x>\bar X^{i}(q,q)\quad\text{and}\quad q=\psi^{c}(x)\iff x=\bar X^{i}(q,q).
\end{equation}
Now fix $i$ such that $q^{i}\leq q^{-i}$ and let
\begin{equation*}
\tau=\inf\{t\geq 0\mid \psi^{c}(X_t)\geq q^{-i}\}.
\end{equation*}
Then, for any $t<\tau$, $\sup_{0\leq s\leq t}\psi^{c}(X_s)\leq q^{-i}$, which implies that $Q_t^{-i}=q^{-i}$ and $Q_t^{i}\leq q^{-i}$. Further, the fact that $\psi^{c}(X_t)<q^{-i}$ for any $t<\tau$ implies that $\bar X^{i}(q^{-i},q^{-i})>X_t=\bar X^{i}(\phi^{i}(X_t,q^{-i}),q^{-i})$ by equivalences \eqref{q<psi} and \eqref{q=phi}, and then necessarily $X_t>\bar X^{i}(\phi^{i}(X_t,q^{-i}),\phi^{i}(X_t,q^{-i}))$ by the fact that $\bar X^{i}$ is strictly increasing in both arguments. Thus, $\phi^{i}(X_t,q^{-i})<\psi^{c}(X_t)$ by the first equivalence in \eqref{q<psi}. Together with the fact that $Q_t^{-i}=q^{-i}$ for all $t<\tau$ as argued before, it follows that
\begin{equation*}
Q_t^{i}=q^{i}\vee\sup_{0\leq s\leq t}\phi^{i}(X_s,q^{-i})=q^{i}\vee\sup_{0\leq s\leq t}\phi^{i}(X_s,Q_s^{-i}),
\end{equation*}
which proves consistency with $\phi^{i}$ for these $t$. Moreover, this implies that $Q_t^{i}\geq\phi^{i}(X_t,q^{-i})$, so that also $q^{-i}\geq\phi^{-i}(X_t,Q_t^{i})$ for all $t<\tau$ by equivalence \eqref{q<phi} and $\bar X^{i}=\bar X^{-i}$. It follows that
\begin{equation*}
Q_t^{-i}=q^{-i}=q^{-i}\vee\sup_{0\leq s\leq t}\phi^{-i}(X_s,Q_s^{i}),
\end{equation*}
which proves also consistency with $\phi^{-i}$ for these $t$.

Next, consider any $t$ such that $\psi^{c}(X_t)\geq q^{-i}$. Then, similarly as before,
%$\bar X^{i}(q^{-i},q^{-i})\leq X_t=\bar X^{i}(\phi^{i}(X_t,q^{-i}),q^{-i})\leq\bar X^{i}(\phi^{i}(X_t,q^{-i}),\phi^{i}(X_t,q^{-i}))$
$\phi^{i}(X_t,q^{-i})\geq\psi^{c}(X_t)$ by equivalences \eqref{q<psi} and \eqref{q=phi}, and the fact that $\bar X^{i}$ is strictly increasing. Thus, by \eqref{c_eql_outcome}, $Q_t^{i}\geq\psi^c(X_t)$ for any such $t$. Since $Q^{i}$ is nondecreasing, it follows that
\begin{equation*}
Q_t^{i}\geq\sup_{0\leq s\leq t}\psi^{c}(X_s)\geq q^{-i}
\end{equation*}
as soon as there is any $s\leq t$ such that $\psi^{c}(X_s)\geq q^{-i}$, which is the case for any $t\geq\tau$ by equivalences \eqref{q<psi} and continuity of $X_t$. The same inequalities hold true for $Q^{-i}$, because also $Q_t^{-i}\geq\psi^c(X_t)$ for any $t$ such that $\psi^{c}(X_t)\geq q^{-i}$ by the same arguments as for $Q^{i}$ and $q^{-i}\geq q^{i}$. For both $Q^{i}$ and $Q^{-i}$, however, the first inequality cannot be strict, because \eqref{c_eql_outcome} and $q^{-i}\geq q^{i}$ imply the reverse weak inequality. Hence, in fact
\begin{equation*}
Q_t^{i}=Q_t^{-i}=\sup_{0\leq s\leq t}\psi^{c}(X_s)\geq q^{-i}
\end{equation*}
for any $t\geq\tau$. To show that this implies consistency with $\phi^{i}$, first note that then $Q_t^{i}=Q_t^{-i}\geq\psi^{c}(X_t)$, so that $X_t\leq\bar X^{i}(Q_t^{i},Q_t^{-i})$ by the first equivalence in \eqref{q<psi} and hence $Q_t^{i}\geq\phi^{i}(X_t,Q_t^{-i})$ by equivalence \eqref{q<phi}. Thus, since the last inequality holds also for all $t<\tau$ by what has already been shown,
\begin{equation*}
Q_t^{i}\geq q^{i}\vee\sup_{0\leq s\leq t}\phi^{i}(X_s,Q_s^{-i})
\end{equation*}
for all $t\geq\tau$ by the fact that $Q^{i}$ is nondecreasing. For consistency with $\phi^{i}$, this inequality must not be strict. To show that the reverse weak inequality holds, note that $\psi^{c}(X_s)$ attains a maximum on any compact interval $[0,t]$, because $X_s$ is continuous in $s$, and since $\psi^{c}(x)$ is a continuous function by being the inverse of the strictly monotone function $f(q)=\bar X^{i}(q,q)$. Thus, let $s\leq t$ be such that $Q_t^{i}=\psi^{c}(X_s)$. Then, by the definition of $\tau$, $t\geq\tau$ implies that necessarily also $s\geq\tau$. Moreover, since $Q^{i}$ is nondecreasing and $Q_s^{i}\geq\psi^{c}(X_s)$ by $s\geq\tau$, already $Q_s^{i}=\psi^{c}(X_s)$. Thus, by the second equivalence in \eqref{q<psi}, $X_s=\bar X^{i}(Q_s^{i},Q_s^{i})$, where $Q_s^{i}=Q_s^{-i}$ by $s\geq\tau$. It follows that $Q_s^{i}=\phi^{i}(X_s,Q_s^{-i})$ by equivalence \eqref{q=phi}, and since $Q_t^{i}=Q_s^{i}$, this shows that indeed
\begin{equation*}
Q_t^{i}\leq\sup_{0\leq s\leq t}\phi^{i}(X_s,Q_s^{-i})\leq q^{i}\vee\sup_{0\leq s\leq t}\phi^{i}(X_s,Q_s^{-i}).
\end{equation*}
The same arguments and the fact that $\bar X^{i}=\bar X^{-i}$ yield also consistency with $\phi^{-i}$. Thus, $(Q^{1},Q^{2})$ is indeed an outcome from $(\phi^{1},\phi^{2})$.
\end{proof}

\begin{lemma}\label{lem:Vc_q-i}
Fix any $c\geq 0$ and let $\bar X^{-i}$ be given by \eqref{barXc}. Then the function $V^{c}$ given by \eqref{Vc}
%\begin{flalign*}%\label{Vc}
%&& V^{c}(x,q^{i},q^{-i})&=\begin{cases} 
%\dfrac{xP(q^{i},q^{-i})q^{i}}{r-\mu}+B(q^{i},q^{-i})x^{\beta}
% &\text{if }x\leq\bar X^{-i}(q^{-i},q^{i}), \\[10pt] 
% V^{c}(x,\phi^{i}(x,q^{-i}),q^{-i})-\phi^{i}(x,q^{-i})+q^{i} &\text{if }x>\bar X^{-i}(q^{-i},q^{i}), 
%\end{cases} && \\[10pt]
%%&& & && \\
%&\text{where}& B(q^{i},q^{-i})&=-\int_{q^{i}}^{\infty}\left(1-\bar X^{-i}(q,q^{-i})\dfrac{P'(q,q^{-i})q+P(q,q^{-i})}{r-\mu}\right)\left(\bar X^{-i}(q,q^{-i})\right)^{-\beta}\,dq. &\hphantom{\text{where}}& \nonumber
%\end{flalign*}
is well defined and has the property that $V_{q^{-i}}^{c}=0$ whenever $x=\bar X^{-i}$ and $q^{i}\geq q^{-i}$.
\end{lemma}

\begin{proof}%[Proof of Lemma~\ref{lem:Vc_q-i}]
First, to see that the function $B(q^{i},q^{-i})$ is well defined by the integral and finite, denote the integrand by $B_{q^{i}}$. Since
\begin{equation*}
0\leq P'q+P=\left(\frac{\gamma-1}{\gamma}q+q^{-i}\right)\left(q+q^{-i}\right)^{-\frac{1}{\gamma}-1}\leq \left(q+q^{-i}\right)^{-\frac{1}{\gamma}}
\end{equation*}
and
\begin{equation*}
0\leq \bar X^{-i}=\left(p^{*}+\frac{c}{q\vee q^{-i}}\right)\left(q+q^{-i}\right)^{\frac{1}{\gamma}}\leq p^{*}\frac{2\gamma}{2\gamma-1}\left(q+q^{-i}\right)^{\frac{1}{\gamma}}
\end{equation*}
for all $q\geq q^{i}\geq c(2\gamma-1)/p^{*}>0$, then
\begin{equation*}
\absd{B_{q^{i}}}\leq \left(1+p^{*}\frac{2\gamma}{2\gamma-1}\right)\left(\bar X^{-i}\right)^{-\beta},
\end{equation*}
where moreover
\begin{equation*}
\left(\bar X^{-i}\right)^{-\beta}\leq \left(p^{*}\right)^{-\beta}\left(q+q^{-i}\right)^{-\frac{\beta}{\gamma}},
\end{equation*}
which is clearly integrable on $[q^{i},\infty)$ by $q^{i}>0$ and $\beta>\gamma$.

Next, since $V_{q^{-i}}^{c}$ is of the form $A_{q^{-i}}x+B_{q^{-i}}x^{\beta}$, the property that $V_{q^{-i}}^{c}=0$ for $x=\bar X^{-i}$ means
\begin{equation*}
B_{q^{-i}}=-A_{q^{-i}}\left(\bar X^{i}\right)^{1-\beta}.
\end{equation*}
Denote the right-hand side by $\tilde B_{q^{-i}}$. Then $\lim_{q^{i}\to\infty}\tilde B_{q^{-i}}=0$, because $\lim_{q^{i}\to\infty}\bar X^{i}=\infty$, $\beta>1$, and
\begin{equation*}
-A_{q^{-i}}=-\frac{P'q^{i}}{r-\mu}=\frac{(q^{i}+q^{-i})^{-\frac{1}{\gamma}}}{\gamma(r-\mu)}\frac{q^{i}}{q^{i}+q^{-i}},
\end{equation*}
which also vanishes as $q^{i}\to\infty$ since $\gamma>1$. Thus, it is possible to use the representation 
\begin{equation*}
\tilde B_{q^{-i}}=-\int_{q^{i}}^{\infty}\tilde B_{q^{-i}q^{i}}(q,q^{-i})\,dq
\end{equation*}
and verify that $B_{q^{i}q^{-i}}(q,q^{-i})=\tilde B_{q^{-i}q^{i}}(q,q^{-i})$ for all $q\geq q^{i}\geq q^{-i}$. Since
\begin{equation*}
\tilde B_{q^{-i}q^{i}}=-\left(\bar X^{-i}\right)^{-\beta}\left(\frac{P''q^{i}+P'}{r-\mu}\bar X^{-i}+(1-\beta)\frac{P'q^{i}}{r-\mu}\bar X_{q^{i}}^{-i}\right)
\end{equation*}
and
\begin{equation*}
B_{q^{i}q^{-i}}=-\left(\bar X^{-i}\right)^{-\beta}\left(\beta\left(\bar X^{-i}\right)^{-1}\bar X_{q^{-i}}^{-i}+(1-\beta)\bar X_{q^{-i}}^{-i}\frac{P'q^{i}+P}{r-\mu}+\bar X^{-i}\frac{P''q^{i}+P'}{r-\mu}\right),
\end{equation*}
these are equal if and only if
\begin{equation*}
P'q^{i}\bar X_{q^{i}}^{-i}+\frac{(r-\mu)\beta}{\beta-1}\left(\bar X^{-i}\right)^{-1}\bar X_{q^{-i}}^{-i}-\bar X_{q^{-i}}^{-i}\left(P'q^{i}+P\right)=0.
\end{equation*}
Using the definitions of $p^{*}$ and $P$, this is equivalent to
\begin{equation*}
-\frac{1}{\gamma}q^{i}\bar X_{q^{i}}^{-i}+p^{*}\left(\bar X^{-i}\right)^{-1}\bar X_{q^{-i}}^{-i}\left(q^{i}+q^{-i}\right)^{\frac{1}{\gamma}+1}-\bar X_{q^{-i}}^{-i}\left(\frac{\gamma-1}{\gamma}q^{i}+q^{-i}\right)=0.
\end{equation*}
Given the specification $\bar X^{-i}=(p^{*}+c/q^{i})(q^{i}+q^{-i})^{1/\gamma}$ (for $q^{i}\geq q^{-i}$), it is easy to verify that the last equation indeed holds.
\end{proof}

\begin{lemma}\label{lem:Vc}
Let $(\phi^{1},\phi^{2})$ be a pair of reflection strategies with boundaries given by \eqref{barXc} for some $c\geq 0$, and consider any state such that $q^{1},q^{2}\geq c(2\gamma-1)/p^{*}$. Then the outcome given by \eqref{c_eql_outcome}
%\begin{equation*}%\label{c_eql_outcome}
%Q_t^{i}=q^{i}\vee\sup_{0\leq s\leq t}\left(\phi^{i}(X_s,q^{-i})\wedge \psi^{c}(X_s)\right),
%\end{equation*}
%where
%\begin{equation*}%\label{psi}
%\psi^{c}(x)=\inf\{q\geq c(2\gamma-1)/p^{*}\mid\bar X^{i}(q,q)\geq x\}.
%\end{equation*}
is optimal for each firm $i$ among all outcomes that are consistent with $\phi^{-i}$, and firm $i$'s payoff from this outcome is equal to the value of the function $V^{c}$ defined in \eqref{Vc}.
\end{lemma}

\begin{proof}%[Proof of Lemma~\ref{lem:Vc}]
The given outcome is indeed admissible and consistent with both strategies by Lemma~\ref{lem:c_eql_outcome}. Hence, it is possible to apply Theorem~\ref{thm:ver} using the given strategies and the function $V=V^{c}$, which is indeed well defined by Lemma~\ref{lem:Vc_q-i}.

Then conditions \ref{V_PDE} and \ref{V_qi=1} hold by construction, and condition \ref{V_q-i=0} is void by $\bar X^{i}=\bar X^{-i}$. To show that the implication in \eqref{V_q-i_0} holds, it is enough to show that $\phi^{i}(x,q)\geq q$ for all relevant $q$, because $\{q^{i}=\phi^{i}(x,q^{-i})\}=\{x=\bar X^{i}(q^{i},q^{-i})\}$ by equivalence \eqref{q=phi}, and $V_{q^{-i}}=0$ on $\{x=\bar X^{i}(q^{i},q^{-i})\}\cap\{q^{i}\geq q^{-i}\}$ by Lemma~\ref{lem:Vc_q-i}. And indeed, whenever $q^{-i}<q<Q_0^{-i}$, then $q<\psi^{c}(X_0)=\psi^{c}(x)$, so that $\bar X^{i}(q,q)<x$ by the first equivalence in \eqref{q<psi}, which implies that $\phi^{i}(x,q)>q$ by equivalence \eqref{q<phi}.

Next, to show that also the implication in \eqref{V_q-i_t} holds, using again the fact that $V_{q^{-i}}=0$ on $\{x=\bar X^{i}(q^{i},q^{-i})\}\cap\{q^{i}\geq q^{-i}\}$, it is enough to show that $dQ_t^{-i}=0$ on $\{X_t=\bar X^{i}(Q_t^{i},Q_t^{-i})\}\cap\{Q_t^{-i}>Q_t^{i}\}$. For the latter set, the fact that $\bar X^{i}$ is strictly increasing in $q^{i}$ implies that $\bar X^{i}(Q_t^{-i},Q_t^{-i})>X_t$, so that $Q_t^{-i}>\psi^{c}(X_t)$ by the equivalences in \eqref{q<psi}. Then, by continuity of $X_t$, it follows that indeed $dQ_t^{-i}=0$.
 
To prove optimality, the integrability requirements \eqref{Vintegrable} and \eqref{limV=0} need to be verified for any admissible outcome that is consistent with $\phi^{-i}$, which will then also cover the given one and complete the proof that firm $i$'s payoff from this outcome is equal to the value of $V^{c}$.

To do so, it will be enough to show that $\abs{V^{c}}$ is bounded by a linear function of $q^{i}+q^{-i}$, because Lemma~\ref{lem:Qintegrable} implies that \eqref{Qintegrable} and \eqref{limQ=0} hold for $Q_T=Q_T^{i}+Q_T^{-i}$ by finiteness of the investment cost for $Q^{i}$ and $Q^{-i}$. This will then yield \eqref{Vintegrable} (by the fact that $Q=Q^{i}+Q^{-i}$ is nondecreasing) and \eqref{limV=0}. Since $X_t\leq\bar X^{-i}(Q_t^{-i},Q_t^{i})$ for all $t\geq 0$ by Lemma~\ref{lem:reflection_outcome}, it is enough to consider only $x\leq\bar X^{-i}$ in \eqref{Vc} to derive the bound on $\abs{V^{c}}$. Then, since $q^{i}\vee q^{-i}\geq c(2\gamma-1)/p^{*}$ by hypothesis,
\begin{equation}\label{x<barXc}
0\leq xP\leq p^{*}\frac{2\gamma}{2\gamma-1}.
\end{equation}
This implies the following bound for the first summand in \eqref{Vc}:
\begin{equation*}
\absd{\frac{xPq^{i}}{r-\mu}}\leq\frac{p^{*}}{r-\mu}\frac{2\gamma}{2\gamma-1}q^{i}\leq\frac{p^{*}}{r-\mu}\frac{2\gamma}{2\gamma-1}\left(q^{i}+q^{-i}\right).
\end{equation*}
For the second summand $Bx^{\beta}$, consider the partial derivative $B_{q^{i}}$. Since $\bar X^{i}\geq p^{*}/P$ by $c\geq 0$, 
\begin{equation*}
B_{q^{i}}\leq\left(\bar X^{i}\right)^{-\beta}\leq\left(\frac{P}{p^{*}}\right)^{\beta}.
\end{equation*}
To obtain a lower bound, note that by $P'<0$ and $\bar X^{i}\geq p^{*}/P$,
\begin{equation*}
B_{q^{i}}\geq-\frac{P}{r-\mu}\left(\bar X^{i}\right)^{1-\beta}\geq-\frac{P}{r-\mu}\left(\frac{P}{p^{*}}\right)^{\beta-1}=-\frac{p^{*}}{r-\mu}\left(\frac{P}{p^{*}}\right)^{\beta}=-\frac{\beta}{\beta-1}\left(\frac{P}{p^{*}}\right)^{\beta}.
\end{equation*}
Combining the upper and lower bounds for $B_{q^{i}}$ and integrating $P^{\beta}=(q^{i}+q^{-i})^{-\beta/\gamma}$ yields
%\begin{equation*}
%\absd{B_{q^{i}}}\leq\frac{\beta}{\beta-1}\left(\frac{P}{p^{*}}\right)^{\beta},
%\end{equation*}
%which implies that
\begin{equation*}
\absd{B}\leq\frac{\beta}{\beta-1}\left(\frac{1}{p^{*}}\right)^{\beta}\frac{\gamma}{\beta-\gamma}\left(q^{i}+q^{-i}\right)^{\frac{\gamma-\beta}{\gamma}}=\frac{\beta}{\beta-1}\left(\frac{P}{p^{*}}\right)^{\beta}\frac{\gamma}{\beta-\gamma}\left(q^{i}+q^{-i}\right).
\end{equation*}
Together with \eqref{x<barXc}, this finally implies the following bound for the second summand in \eqref{Vc}:
\begin{equation*}
\absd{Bx^{\beta}}\leq\frac{\beta}{\beta-1}\left(\frac{xP}{p^{*}}\right)^{\beta}\frac{\gamma}{\beta-\gamma}\left(q^{i}+q^{-i}\right)\leq\frac{\beta}{\beta-1}\left(\frac{2\gamma}{2\gamma-1}\right)^{\beta}\frac{\gamma}{\beta-\gamma}\left(q^{i}+q^{-i}\right),
\end{equation*}

Now all that remains to prove is that $V=V^{c}$ has the three additional properties in Theorem~\ref{thm:ver}. Condition \ref{V_PDE_leq} holds trivially by condition \ref{V_PDE} and $\bar X^{i}=\bar X^{-i}$. 

Concerning condition \ref{V_qi<1}, note that $V_{q^{i}}^{c}$ is of the form $A_{q^{i}}x+B_{q^{i}}x^{\beta}$, and by construction $V_{q^{i}}^{c}=1$ for $x=\bar X^{i}$. Hence, the condition can be written as
\begin{equation*}
\forall x\leq\bar X^{i}\colon\quad\left(1-A_{q^{i}}x\right)x^{-\beta}\geq B_{q^{i}}=\left(1-A_{q^{i}}\bar X^{i}\right)\left(\bar X^{i}\right)^{-\beta}.
\end{equation*}
The left-hand side of this inequality, where $A_{q^{i}}>0$ by $P(q)=q^{-1/\gamma}$, is strictly decreasing in $x$ up to $x=\beta/((\beta-1)A_{q^{i}})$ and then strictly increasing. Therefore, condition \ref{V_qi<1} holds if (and only if) 
\begin{flalign*}
&& &\mathclap{\bar X^{i}\leq\frac{\beta}{\beta-1}\left(A_{q^{i}}\right)^{-1}} && \\
&\mathrlap{\qquad\iff} &&\mathclap{\left(p^{*}+\frac{c}{q^{i}\vee q^{-i}}\right)\left(P\right)^{-1}\leq\frac{\beta}{\beta-1}(r-\mu)\left(P+P'q^{i}\right)^{-1}=p^{*}\frac{q^{i}+q^{-i}}{\frac{\gamma-1}{\gamma}q^{i}+q^{-i}}\left(P\right)^{-1}} && \\
&\mathrlap{\qquad\iff} &&\mathclap{c\frac{\frac{\gamma-1}{\gamma}q^{i}+q^{-i}}{q^{i}\vee q^{-i}}\leq\frac{p^{*}}{\gamma}q^{i}.} &&
\end{flalign*}
For any fixed $q^{i}$, the left-hand side of the last inequality attains a maximum at $q^{-i}=q^{i}$, and then the inequality becomes equivalent to $c(2\gamma-1)/p^{*}\leq q^{i}$, which is indeed true by hypothesis. 

Condition \ref{V_q-i<0} holds (with equality) for $q^{i}\geq q^{-i}$ by Lemma~\ref{lem:Vc_q-i}, so consider $q^{i}\leq q^{-i}$. Since also $V_{q^{-i}}^{c}$ is of the form $A_{q^{-i}}x+B_{q^{-i}}x^{\beta}$, the condition can be written as
\begin{equation*}
B_{q^{-i}}\leq -A_{q^{-i}}\left(\bar X^{i}\right)^{1-\beta}.
\end{equation*}
Denote the right-hand side by $\tilde B_{q^{-i}}$. Because equality holds for $q^{i}=q^{-i}$, the condition is for any $q^{i}<q^{-i}$ equivalent to
\begin{equation*}
\int_{q^{i}}^{q^{-i}}B_{q^{-i}q^{i}}(q,q^{-i})\,dq \geq \int_{q^{i}}^{q^{-i}}\tilde B_{q^{-i}q^{i}}(q,q^{-i})\,dq.
\end{equation*}
Using the facts that
\begin{equation*}
-A_{q^{-i}}=-\frac{P'q^{i}}{r-\mu}=\frac{(q^{i}+q^{-i})^{-\frac{1}{\gamma}-1}q^{i}}{\gamma(r-\mu)}
\end{equation*}
and that $\bar X^{i}$ is given by \eqref{barXc}, a lengthy calculation yields that indeed $B_{q^{-i}q^{i}}\geq \tilde B_{q^{-i}q^{i}}$ for all $q^{i}<q^{-i}$ if and only if $q^{-i}\geq c(\gamma-1)/p^{*}$. Since the latter is implied by the hypothesis that $q^{-i}\geq c(2\gamma-1)/p^{*}$, this completes the proof.
\end{proof}

%\nocite{*}
%\bibliography{jhs}
 \newcommand{\noop}[1]{}

\end{document}